\documentclass[journal,a4paper,10pt]{IEEEtran}  

\usepackage[utf8]{inputenc} 
\usepackage[T1]{fontenc}
\usepackage{url}
\usepackage{ifthen}
\usepackage{cite}
\usepackage[cmex10]{amsmath} 
\usepackage{amssymb,mathtools}
\usepackage{multirow}
\usepackage{graphicx}
\usepackage{amsthm}

\usepackage{mleftright}

\usepackage{framed}
\usepackage[framemethod=TikZ]{mdframed}

\usepackage{tikz}
\usetikzlibrary{tikzmark}
\usetikzlibrary{shapes,snakes}
\usetikzlibrary{fit,backgrounds}
\usetikzlibrary{decorations.pathreplacing,angles,quotes}
\usetikzlibrary{patterns}
\usetikzlibrary{calc}
\usetikzlibrary{positioning}
\tikzstyle{block} = [rectangle, draw, 
    minimum width=4em, text centered, rounded corners, minimum height=1em]
\tikzstyle{network} = [rectangle, draw, 
    minimum width=10em, text centered, rounded corners, minimum height=6em]
\tikzstyle{rec} = [rectangle, draw]
\tikzstyle{line} = [draw, -latex, thick]

\usepackage{color}

\newtheorem{lemm}{Lemma}[section]

\newtheorem{theo}{Theorem}[section]
\newtheorem{prop}{Proposition}[section]

\newcommand{\p}{\hspace*{0.5em}}

\DeclareMathOperator*{\Tr}{\text{Tr}}

\DeclareMathOperator{\Ima}{Im}
\DeclareMathOperator{\tr}{tr}
\DeclareMathOperator{\pr}{Pr}

\DeclareMathOperator{\cH}{{\mathcal{H}}}
\DeclareMathOperator{\rank}{rank}

\newcommand\kb[1]{| #1 \rangle\langle #1 |}

\newcommand\zerom[2]{\mathbf{0}_{#1,#2}}

\newcommand\lefto[1]{\mathsf{L}(#1)}
\newcommand\righto[1]{\mathsf{R}(#1)}
\newcommand\leftop[1]{\mathsf{L'}(#1)}
\newcommand\rightop[1]{\mathsf{R'}(#1)}

\DeclarePairedDelimiter\px{\{}{\}}
\DeclarePairedDelimiter\paren{(}{)}
\DeclarePairedDelimiter\bparen{[}{]}
\DeclarePairedDelimiter\ceil{\lceil}{\rceil}
\DeclarePairedDelimiter\floor{\lfloor}{\rfloor}

\DeclarePairedDelimiter\bra{\langle}{\rvert}
\DeclarePairedDelimiter\ket{\lvert}{\rangle}
\DeclarePairedDelimiterX\braket[2]{\langle}{\rangle}{#1 \delimsize\vert #2}
\DeclarePairedDelimiterX\ketbra[2]{\delimsize\vert}{\delimsize\vert}{#1 \rangle\langle #2}

\DeclarePairedDelimiter\bitk{\rvert}{\rangle_{b}}

\DeclarePairedDelimiter\phasek{\lvert}{\rangle_p }

\def\cHp{\mathcal{H}'}
\def\mm{m_0}
\def\vsec{S}
\def\FF{\mathbb{F}}

\def\Q{q'}
\def\Code{\mathsf{C}}

\begin{document}
\title{
Secure Quantum Network Code without Classical Communication 
}


\author{%
Seunghoan Song
 and~Masahito~Hayashi,~\IEEEmembership{Fellow,~IEEE}
\thanks{S. Song and M. Hayashi are with Graduate school of Mathematics, Nagoya University, Nagoya, Japan
(e-mail: m17021a@math.nagoya-u.ac.jp, masahito@math.nagoya-u.ac.jp).}
\thanks{M. Hayashi is also with 
Shenzhen Institute for Quantum Science and Engineering, Southern University of Science and Technology and the Centre for Quantum Technologies, 
National University of Singapore, Singapore.}
\thanks{This paper was presented in part at Proceedings of 2018 IEEE Information Theory Workshop, Guanzhou, China.}
}


\maketitle

\begin{abstract}  
We consider the secure quantum communication over a network with the presence of a malicious adversary who can eavesdrop and contaminate the states.
{
    The network consists of noiseless quantum channels with the unit capacity and the nodes which applies noiseless quantum operations.
    As the main result, when the maximum number $m_1$ of the attacked channels over the entire network uses is less than a half of the network transmission rate $m_0$ (i.e., $m_1 < m_0/2$), our code implements secret and correctable quantum communication of the rate $m_0-2m_1$ by using the network asymptotic number of times.
}
  Our code is universal in the sense that the code is constructed without the knowledge of the specific node operations and the network topology,
  but instead, every node operation is constrained to the application of an invertible matrix to {the} basis states.
  Moreover, our code requires no classical communication.
  Our code can be thought of as a generalization of the quantum secret sharing.
\end{abstract}


\section{Introduction}

\IEEEPARstart{N}{etwork} coding is a coding method, addressed first by Ahlswede et al. \cite{ACLY},
that allows network nodes to manipulate information {packets} before forwarding.
As a quantum analog, quantum network coding considers sending quantum states
through a network which consists of noiseless quantum channels and nodes performing quantum operations.
Since it was first discussed by Hayashi et al. \cite{Hayashi2007},
many other papers \cite{PhysRevA.76.040301,Kobayashi2009,Kobayashi2010, Kobayashi2011,Leung2010,OKH17,OKH17-2} have studied quantum network codes.

Classical network codes with security have been studied by two different methods.
One method is to combine the network node controls and an end-to-end code.
In this method, the sender and receiver know the network topology, control the node operations, and construct an end-to-end code between them.
The use of the end-to-end code is important because it generates the redundancy which is necessary for the security guarantee.
By this method, Cai and Yeung \cite{Cai2002} first devised a classical network code which guarantees the secrecy of the communication.
Secure classical network codes by this method have been further studied in \cite{FMSS04, RS07}.

The other method for secure classical network codes is to use only an end-to-end code without controlling node operations.
In this method, the node operations are not directly controlled but constrained,
and an end-to-end code is constructed with the knowledge of the constraints without specific knowledge of the underlying node operations and the network topology.
Although the codes \cite{Jaggi2008,Matsumoto2011a,YSJL14,HOKC17} by this method do not control the node operations, which differs from the original definition of the network code in \cite{ACLY},
these codes are also called network codes.
By this method, Jaggi et al. \cite{Jaggi2008} constructed a classical network code with asymptotic error correctability.
In the paper \cite{Jaggi2008}, all node operations are not controlled but constrained to be linear operations, and the code is universal 
in the sense that the code is constructed independently of the network topology and the particular node operations.
When the transmission rate $m_0$ of the network and the maximum rate $m_1$ of the malicious injection satisfy $m_1< m_0$,
the code in \cite{Jaggi2008} achieves the correctability with the rate $m_0-m_1$ by asymptotic $n$ uses of the network.
Furthermore, Hayashi et al. \cite{HOKC17} extended the result in \cite{Jaggi2008} so that the secrecy is also guaranteed:
when previously defined $m_0$, $m_1$, and the information leakage rate $m_2$ satisfy $m_1+m_2 <m_0$,
the classical network code in \cite{HOKC17}
achieves the secrecy and the correctability with the rate $m_0-m_1-m_2$ by asymptotic $n$ uses of the network.

On the other hand, 
secure quantum network codes have been designed by Owari et al. \cite{OKH17} and Kato et al. \cite{OKH17-2}.
However, 
the codes in \cite{OKH17,OKH17-2} only keep secrecy from the malicious adversary
but do not guarantee the correctness of the transmitted state if there is an attack.
Moreover, 
this code depends on the network topology and requires classical communication.

In this paper, 
to resolve these problems
and
as a natural quantum extension of the secure classical network codes \cite{Jaggi2008,HOKC17}, 
we present a quantum network code which is secret and correctable.
Since we take a similar method to \cite{Jaggi2008,HOKC17}, our code 
{consists only of an end-to-end code without node operation controls and}
transmits a state by multiple $n$ uses of the quantum network.
When the network transmission rate is $m_0$ and the maximum number $m_1$ of the attacked channels satisfy $m_1 < m_0/2$,
{our code transmits quantum information of the rate $m_0-2m_1$ with high fidelity by asymptotic $n$ uses of the network.}
Since {the high fidelity} of the transmitted quantum state guarantees the secrecy of the transmission \cite{Schumacher96},
the secrecy of our code is guaranteed.

There are several notable properties in our code.
First, our code is universal
in the sense that the code construction does not depend on the network topology and the particular node operations.
{Instead, we place two constraints on the network topology and node operations.
That is, at every node, the number of incoming edges is the same as the number of outgoing edges, 
and, similarly to \cite{Jaggi2008,HOKC17} but differently from \cite{OKH17,OKH17-2}, 
every node operation is the application of an invertible matrix to basis states.
Then, our code is constructed 
by using the constraints but 
without any knowledge of the network topology and operations.}
{Secondly, our code can be constructed without any classical communication.
Though a negligible rate secret shared randomness is necessary for our code construction, 
we attach a subprotocol in order for sharing the randomness by use of the quantum network,
and therefore no classical communication or no assumption of shared randomness is needed.}
Thirdly, our code is secure from any malicious operation on $m_1$ channels
if $m_1 < m_0/2$.
That is, when $m_1 < m_0/2$, our code is secure from the strongest eavesdropper 
who knows the network topology and the network operations, keeps  classical
information extracted from the wiretapped states, and applies quantum operations on the attacking channels adaptively by her wiretapped information.
Fourthly, 
when the network consists of parallel $m_0$ quantum channels,
our code can be thought of as 
an error-tolerant quantum secret sharing \cite{CGS05}.

The rest of this paper is organized as follows.
Section \ref{sec:model} formally describes the quantum network and the attack model.
Section \ref{sec:main_result} presents two main results of the paper, and compares our quantum network code with {the quantum maximum distance separable (MDS) codes and quantum secret sharing}.
Based on the preliminaries in Section \ref{sec:prelim},
Section \ref{sec:protocol} constructs our code when a negligible rate secret shared randomness is assumed. 
Section \ref{sec:analysis} evaluates the performance of the code 
and shows that the entanglement fidelity of the code protocol is bounded by the sum of two error probabilities, called bit error probability and phase error probability.
Section \ref{sec:errors} derives upper bounds of the bit error probability and phase error probability, respectively.
Section \ref{sec:SQNC} constructs our code without assuming any negligible rate secret shared randomness.
Section \ref{sec:secrecy} analyzes the secrecy of our code.
Section \ref{sec:conclusion} is the conclusion of the paper.

\section{Quantum Network and Attack Model}   \label{sec:model}

We give the formal description of our quantum network which is defined as a natural quantum extension of a classical network.
The notations in the network and attack model are summarized in Table \ref{tab:rates}, and an example of the quantum network is given in Fig.~\ref{fig:1}.

\begin{table} 
\renewcommand{\arraystretch}{1.2}
    \centering
    \caption{Summary of Notations} \label{tab:rates}
    \begin{tabular}[ht]{|l|l|}
    \hline
    $m_0$ & Network transmission rate without attack \\ \hline
    $m_1$ ($< m_0/2$) & Maximum number of attacked channels \\ \hline
    $m_a$ ($\leq m_1$) & Number of attacked channels \\ \hline
    $\cH$ & Unit quantum system \\ \hline
    $q$ & Dimension of $\cH$ (prime power) \\ \hline
    $n$ & Block-length \\ \hline
    $\mathcal{F}$ & Network structure \\ \hline     
    $S_n$ & Strategy of malicious attack \\ \hline
    $\Gamma[\mathcal{F}^{n},S_n]$ & Network operation \\ \hline
    $\Code_n$ & Quantum network code \\ \hline
    $\cH_{\mathrm{code}}^{(n)}$ & Code space \\ \hline
    $\Lambda_n = \Lambda[\mathsf{C}_{n},\mathcal{F}^{n},S_n]$ & Averaged protocol by code randomness  \\ \hline
    $\cHp$ & Extended unit quantum system \\ \hline
    $\alpha$ & Dimension of extension \\ \hline
    $\Q=q^{\alpha}$ & Dimension of $\cHp$  \\ \hline
    $n'$ & Block-length with respect to $\cHp$ \\ \hline
    $|x \rangle_b$ ($x\in\mathbb{F}_q$ ($\mathbb{F}_{\Q}$)) & Bit basis element of $\cH$ ($\cHp$) \\ \hline
    $|z \rangle_p$ ($z\in\mathbb{F}_q$ ($\mathbb{F}_{\Q}$)) & Phase basis element of $\cH$ ($\cHp$) \\ \hline
    \end{tabular}
\end{table}

\subsection{Network structure and transmission}

We consider the network described by a directed acyclic graph $G_{m_0}=(V,E)$ 
where $V$ is the set of nodes (vertices) and $E$ is the set of channels (edges).
The network $G_{m_0}$ has 
one source node $v_0$, intermediate nodes $v_1$, \ldots, $v_c$ ($c:=|V|-2$), and one sink node $v_{c+1}$,
where the subscript represents the order of the information conversion.
The source node $v_0$ and the sink node $v_{c+1}$ have $m_0$ outgoing and incoming channels, respectively,
and each intermediate node $v_t$ has the same number {$k_t\in\{1,\ldots, m_0\}$} of incoming and outgoing channels.
For convenience, we define $k_0=k_{c+1} := m_0$.

The transmission on the network $G_{m_0}$ is described as follows.
{Each channel transmits information noiselessly unless the channel is attacked, and each node applies an information conversion noiselessly at any time.}
At time $0$, the source node transmits the input information along the $m_0$ outgoing channels.
At time {$t\in\{1,\ldots,c\}$}, the node $v_t$ applies an information conversion to the information from the $k_t$ incoming channels, and outputs the conversion outcome along the $k_t$ outgoing channels.
At time $c+1$, the sink node receives the output information from the $m_0$ incoming channels.
The detailed constraints of the transmitted information and information conversion are described in the following subsections.

{The $m_0$ outgoing channels of the source node are numbered from $1$ to $m_0$, and 
after the conversion in the node $v_t$, the assigned numbers are changed from $k_t$ incoming channels to $k_t$ outgoing channels {deterministically}.

\subsection{Classical network}
To explain our model of the quantum network, we first consider the classical network.
Every single use of a channel transmits one symbol of the finite field $\FF_q$ of order $q$.
Hence, the information {at each time} is described by the vector space $\FF_q^{m_0}$.
We assume that the information conversion at each intermediate node is an invertible linear operation.
That is, the information conversion at each intermediate node $v_t$ is written as an invertible $k_t \times k_t$ matrix $A_t$ acting only on the $k_t$ components of the vector space $\FF_q^{m_0}$.
Therefore, combining all the conversions, 
the relation between the input information $x\in\FF_q^{m_0}$ and the output information $y\in\FF_q^{m_0}$ can be characterized by an invertible $m_0 \times m_0$ matrix $K$ as
\begin{align}
y= K x.
\end{align}

\begin{figure}[t]
\centering
\begin{tikzpicture}[transform shape, node distance = 5cm, auto, every node/.style={outer sep=0}]
    \node [circle,draw,inner sep = 0.2em] (source) {$v_0$};
    \node [below left=-0.1em and -1em of source] {Source};
    \node [circle,draw, above right=4em and 4em of source,inner sep = 0.2em] (node1) {$v_1$};
    \node [above=-0.1em of node1] {$\lefto{A_1}$};
    \node [circle,draw, below right=4em and 6em of source,inner sep = 0.2em] (node2) {$v_2$};
    \node [below=-0.1em of node2] {$\lefto{A_2}$};
    \node [circle,draw, right=10em of source,inner sep = 0.2em] (node3) {$v_3$};
    \node [above=-0.1em of node3] {$\lefto{A_3}$};
    \node [circle,draw, right=18em of source,inner sep = 0.2em] (target) {$v_4$};
    \node [below right=0.1em and -1em of target] {Target};
    
    \path [line] (source) edge [bend left=30] (target);
    \path [line] (source) edge [bend left=25]  (node1);
    \path [line] (source) edge [bend left=0]  (node1);
    
    \path [line] (node1) edge [bend left=25] (target);
    \path [line] (node1) edge (node2);
    
    \path [line] (source) edge [bend left=15]  (node2);
    \path [line] (source) edge [bend right=15]  (node2);
    
    \path [line] (source) edge  (node3);

    \path [line] (node2) edge [bend left=15] (node3);
    \path [line] (node2) edge [bend right=15] (node3);
    \path [line] (node2) edge [bend right=15] (target);
    
    \path [line] (node3) edge [bend left=20] (target);
    \path [line] (node3) edge [bend right=20] (target);
    \path [line] (node3) edge  (target);
\end{tikzpicture}
\begin{align*}
A_1 = \begin{bmatrix}
        1 & 3 \\
        2 & 3
      \end{bmatrix},
\quad
A_2 = \begin{bmatrix}
        1 & 2 & 3 \\
        0 & 1 & 5 \\ 
        5 & 6 & 0 
      \end{bmatrix},
\quad
A_3 = \begin{bmatrix}
        0 & 2 & 2 \\
        1 & 1 & 1 \\ 
        0 & 1 & 2 
      \end{bmatrix}.
\end{align*}
\caption{Quantum network with three intermediate nodes. 
Source and sink nodes have $m_0=6$ outgoing and incoming channels, respectively, and 
each intermediate node has the same number of incoming and outgoing channels.
Each channel transmits $7$-dimensional Hilbert space, i.e., $q=7$, and each intermediate node $v_t$ for $t=1,2,3$ applies $\lefto{A_t}$, where $A_t$ is an invertible matrix over $\mathbb{F}_7$.
} \label{fig:1}
\vspace{1em}

\begin{tikzpicture}[transform shape, node distance = 5cm, auto, every node/.style={outer sep=0}]
    \node [circle,draw,inner sep = 0.2em] (source) {$v_0$};
    \node [below left=-0.1em and -1em of source] {Source};
    \node [circle,draw, above right=4em and 4em of source,inner sep = 0.2em] (node1) {$v_1$};
    \node [above=-0.1em of node1] {$\lefto{A_1}$};
    \node [circle,draw, below right=4em and 6em of source,inner sep = 0.2em] (node2) {$v_2$};
    \node [below=-0.1em of node2] {$\lefto{A_2}$};
    \node [circle,draw, right=10em of source,inner sep = 0.2em] (node3) {$v_3$};
    \node [above=-0.1em of node3] {$\lefto{A_3}$};
    \node [circle,draw, right=18em of source,inner sep = 0.2em] (target) {$v_4$};
    \node [below right=0.1em and -1em of target] {Target};
    
    \path [line] (source) edge [bend left=30] (target);
    
    \path [line] (source) edge [bend left=25, pos=0.8,->,decorate,decoration={snake,amplitude=.4mm,segment length=2mm,post length=1mm}] node {{Attack}} (node1);
    \path [line] (source) edge [bend left=0]  (node1);
    
    \path [line,dashed] (node1) edge [bend left=25] (target);
    \path [line,dashed] (node1) edge (node2);
    
    \path [line] (source) edge [bend left=15]  (node2);
    \path [line] (source) edge [bend right=15]  (node2);
    
    \path [line] (source) edge  (node3);

    \path [line,dashed] (node2) edge [bend left=15] (node3);
    \path [line,dashed] (node2) edge [bend right=15] (node3);
    \path [line,dashed] (node2) edge [bend right=15] (target);
    
    \path [line,dashed] (node3) edge [bend left=20] (target);
    \path [line,dashed] (node3) edge [bend right=20] (target);
    \path [line,dashed] (node3) edge  (target);
\end{tikzpicture}
\caption{Propagation of malicious corruption in quantum network of Fig.~1 when Eve attacks the first channel (zigzagged) of the source node. 
The malicious corruption propagates by node operations along dashed channels.
The target node receives $5$ corrupted unit quantum systems.}   \label{fig:2}
\end{figure}

We extend the above discussion to the case of $n$ network uses, i.e., 
the input and output informations are written as $X=[x_1,\ldots,x_n]\in\mathbb{F}_q^{m_0\times n}$ and $Y=[y_1,\ldots,y_n]\in\mathbb{F}_q^{m_0\times n}$.
We assume that 
every intermediate node $v_t$ applies the invertible matrix $A_t$ {at $n$ times} and
the matrix $A_t$ is not changed during the $n$ transmissions.
In addition, we assume that the inputs $x_1,\ldots,x_n$ are {independently transmitted}, i.e., 
$y_i = Kx_i$ holds for any $i\in\{1,\ldots,n\}$.
Therefore, we have {the relation}
\begin{align}
Y= K X.\label{H1}
\end{align}

Next, we extend more to the case where a malicious adversary Eve attacks $m_a$ ($\leq m_1$) channels, i.e., fixed $m_a$ channels are attacked over $n$ uses of the network (Fig. 2).
Since all the node operations are linear,
there is a linear relation between 
the information on each channel and output information.
That is, there are $m_a$ vectors $w_1, \ldots, w_{m_a}$ in $\FF_q^{m_0}$ satisfying the following condition:
when Eve adds the noise $z_1, \ldots, z_{m_a} \in \FF_q^n$ on the $m_a$ attacked channels,
the relation \eqref{H1} is changed to
\begin{align}
Y
= K X+ \sum_{j=1}^{m_a} w_j z_j^{\top}
= K X+
W Z,\label{H2}
\end{align}
where $W=[w_1, \ldots, w_{m_a}]$ and $Z=[z_1, \ldots, z_{m_a}]^{\top}$.
Here, the vectors $w_1, \ldots, w_{m_a}$
are determined by the network topology and node operations.
For the detail, see \cite[Section 2.2]{OKH17-2}.
Even in the case where Eve chooses the noise $Z$ dependently of the input information $X$, 
the output information $Y$ is always written in the form \eqref{H2}. 

\subsection{Quantum network} \label{subsec:quant_network}
We consider a natural quantum extension of the above classical network.
Every single use of a quantum channel transmits a quantum system $\mathcal{H}$ of dimension $q$ spanned by a basis $\px{|x\rangle_{b} \mid x \in \FF_q}$
{which is called the {\em bit basis}.}
In $n$ uses of the network,
the whole system to be transmitted is written as $\mathcal{H}^{\otimes m_0 \times n} $
spanned by $\px{ | X\rangle_{b} \mid X \in \FF_q^{m_0 \times n}}$.
To describe the node operations, we introduce the following unitary operations:
for an invertible $m \times m $ matrix $A$ and an invertible $n \times n $ matrix $B$,
two unitaries $ \lefto{A}$ and $\righto{B}$ are defined as
\begin{align}
 \lefto{A} &:=\!\!\! \sum_{X\in\FF_q^{m \times n}}\!\!\!\ket{AX}_{bb}\bra{X}, \!\quad
 \righto{B} :=\!\!\! \sum_{X\in\FF_q^{m \times n}}\!\!\!\ket{XB}_{bb}\bra{X}.\!
 \label{eq:chops}
\end{align}
Every node $v_t$ converts the information on the subsystem $ \mathcal{H}^{\otimes k_t \times n}$ 
by applying the unitary $\lefto{A_t}$.
If there is no attack, the operation of the whole network is the application of 
the unitary $\lefto{K} $.

{Next, we introduce Eve's attack model.
Eve attacks fixed $m_a$ ($\leq m_1$) channels over $n$ uses of the network.
Whenever quantum systems are transmitted over the $m_a$ attacked channels,
Eve can perform on the systems any trace preserving and completely positive (TP-CP) maps, measurements defined by positive operator-valued measure (POVM), or both.
We assume that Eve's operations can be adaptive on the previous measurement outcomes and Eve knows the network topology and all node operations.
}

{
Consider the entire network operation with malicious attacks.
When Eve attacks on channels,
the network structure $\mathcal{F}$ is characterized by the network topology $G_{m_0}=(V,E)$, node operations $A=(A_1,\ldots,A_c)$, and the set $E_{\mathrm{att}}\subset E$ of attacked channels, i.e., $\mathcal{F} := (G_{m_0}, A, E_{\mathrm{att}})$.
Given a network structure $\mathcal{F}$, Eve's strategy $S_{n}$ over $n$ network uses determines the TP-CP map of the entire network operation.
Therefore, we denote the entire network operation over $n$ network uses as a TP-CP map 
\begin{align}
\Gamma [\mathcal{F}^n, S_n],    \label{eq:network_operations}
\end{align}
where $\mathcal{F}^n$ denotes the network structure $\mathcal{F}$ is used $n$ times.
As a special case, if $E_{\mathrm{att}} = \emptyset$,
we have $\Gamma [\mathcal{F}^n, S_n] = \lefto{K}\rho\lefto{K}^{\dagger}$.
Moreover, 
we define the set $\zeta_{m_0,m_1}^{(n)}$ of all network structures and strategies 
of transmission rate $m_0$ without attacks, at most $m_1$ attacked channels, and block-length $n$
as
\begin{align}
&\zeta_{m_0,m_1}^{(n)} \nonumber\\
&:= \{(\mathcal{F} , S_n) \mid \mathcal{F}= (G_{m_0}, A, E_{\mathrm{att}}), \ \! m_a\!=\! |E_{\mathrm{att}}| \!\leq\! m_1\}.
 \label{eq:set_network_operations}
\end{align}
}

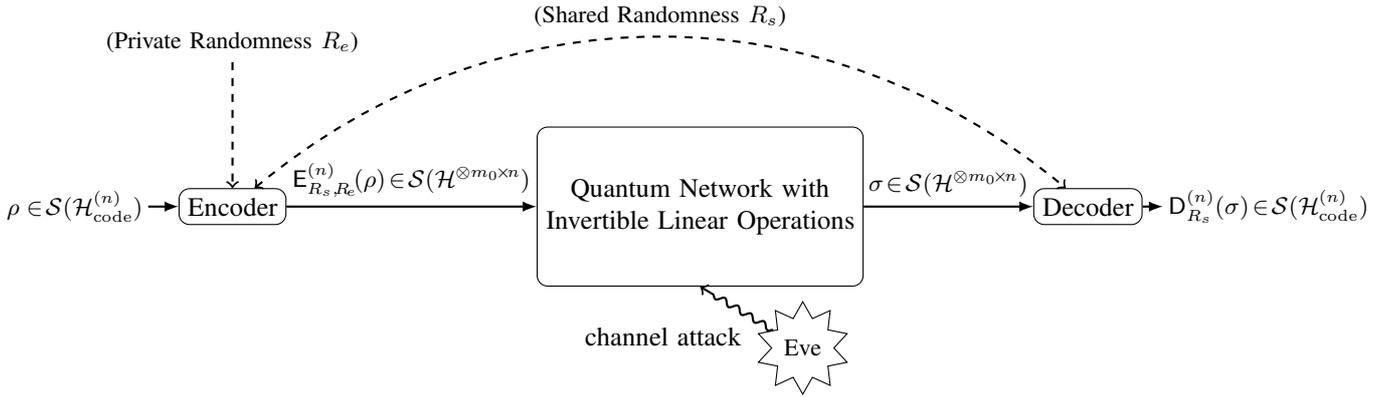
\begin{figure*}[tb]
        \centering
        
\begin{tikzpicture}[scale=0.75, every text node part/.style={align=center},node distance = 5cm, auto]
    \node [network] (channel) {Quantum Network with\\ Invertible Linear Operations};
    \node [block] (encoder) at (-8.2,0) {Encoder};
    \node (pr) at (-8.2,2.9) {\small (Private Randomness $R_e$)};
    \node [block] (decoder) at (6.8,0) {Decoder};
    \node at (-11, 0) (input) {\small $\rho\in\! \mathcal{S}(\cH_{\mathrm{code}}^{(n)})$\!};
    \node at (10, 0) (output) {\small \!$\mathsf{D}^{(n)}_{R_s}(\sigma)\!\in\! \mathcal{S}(\cH_{\mathrm{code}}^{(n)})$} ;
    
    \path [line,dashed] (encoder) edge [<->,bend left=40] node {\small (Shared Randomness $R_s$)} (decoder) ;
    \path [line] (input) -- (encoder);
    \path [line] (encoder) -- node[pos=0.5] {\small $\mathsf{E}_{R_s,\!R_{\!e}}^{(n)}\!(\rho)\!\in\!\mathcal{S}(\cH^{\otimes m_0\!\times\! n})$}(channel);
    \path [line] (channel) -- node[pos=0.5] {\small $\sigma\!\in\!\mathcal{S}(\cH^{\otimes m_0\!\times\! n})$}(decoder);
    \path [line] (decoder) -- (output);
    \path [line,dashed] (pr) edge [->] (encoder);
    
    \node at (1.8,-2.5) [draw,star,star points=10,inner sep=1.3pt] (eve) {\small Eve};
    \path [line] (eve) edge [pos=0.3,->,decorate,decoration={snake,amplitude=.4mm,segment length=2mm,post length=1mm}] node {{channel attack}} (channel.south);
\end{tikzpicture}
        \caption{Protocol with negligible rate secret shared randomness. 
        $\mathcal{S}(\mathcal{H})$ denotes the set of density matrices on the Hilbert space $\mathcal{H}$.} \label{fig1}
\end{figure*}

\section{Main Results} \label{sec:main_result}

In this section, we present the two coding theorems with and without a negligible rate secret shared randomness.
For any quantum network described in Section \ref{sec:model},
our code can be constructed 
only with the knowledge of $m_0$, $m_1$, and $q$, 
but without any specific knowledge of the node operations $\lefto{A_t}$ and the network topology $G_{m_0}$.

\subsection{Main idea in our code construction}

In order to explain the main idea of our code, we briefly introduce the classical network codes in \cite{Jaggi2008,HOKC17}.
In \cite{Jaggi2008,HOKC17}, node operations are restricted to be linear operations. 
Therefore, malicious injections on channels form a subspace in the network output, in the same way as \eqref{H2}.
Then, the codes in \cite{Jaggi2008,HOKC17} find the subspace of injections from the network output with the help of secret shared randomness between the sender and receiver.
Finally, the codes recover the original message from the information not in the subspace of injections.

By the above method of the classical network codes in \cite{Jaggi2008,HOKC17},
our quantum network code is designed in the following way.
Since our quantum network in Section \ref{sec:model} is defined as a natural quantum extension of the classical networks in \cite{Jaggi2008,HOKC17},
we can reduce the correctness of our code to that of two classical network codes which are defined on two bases of quantum systems (in Sections \ref{sec:analysis} and \ref{subsec:Application}).
In this reduction, our quantum network code is sophisticatedly defined so that the two classical network codes are similar to the codes in \cite{Jaggi2008,HOKC17}.
A difficult point in our code construction is that the accessible information from the network output state is restricted since a measurement disturbs the quantum states, whereas the classical network codes \cite{Jaggi2008,HOKC17} have access to all information of the network output.
Our code circumvents this difficulty by attaching to the codeword the ancilla whose measurement outcome contains sufficient information for finding the subspace of injections.

\subsection{Main theorems}
In this subsection, we present two coding theorems with and without a negligible rate secret shared randomness.

Before we state the two coding theorems, we formulate a quantum network code of block-length $n$.
Let $\mathcal{R}_s$ and $\mathcal{R}_e$ be {sets for the secret shared randomness and the private randomness parameters, respectively.}
Let $\cH_{\mathrm{code}}^{(n)}$ be a quantum system called the {\em code space}.
Given $(r_s,r_e)\in\mathcal{R}_s\times\mathcal{R}_e$, an encoder is defined as a TP-CP map $\mathsf{E}_{r_s,r_e}^{(n)}$ from $\cH_{\mathrm{code}}^{(n)}$ to $\cH^{\otimes m_0\times n}$,
and a decoder is defined as a TP-CP map $\mathsf{D}_{r_s}^{(n)}$ from $\cH^{\otimes m_0\times n}$ to $\cH_{\mathrm{code}}^{(n)}$.
{The parameter $r_s$} is assumed to be shared between the encoder and decoder but kept a secret to all others, 
and $r_e$ is a private randomnesses of the encoder.
Then, a quantum network code is defined as 
\begin{align}
\Code_{n}:= \{(\mathsf{E}_{r_s,r_e}^{(n)}, \mathsf{D}_{r_s}^{(n)})\mid (r_s,r_e)\in\mathcal{R}_s\times\mathcal{R}_e \}.
\end{align}

In order to evaluate the performance of a quantum network code $\mathsf{C}_{n}$, we consider the averaged protocol
\begin{align}
&\Lambda[\mathsf{C}_{n},\mathcal{F}^{n},S_n](\rho) \nonumber \\
&:=  \frac{1}{|\mathcal{R}_s\times\mathcal{R}_e|} \sum_{(r_s,r_e)} \mathsf{D}_{r_s}^{(n)} \circ \Gamma[\mathcal{F}^n, S_n] \circ \mathsf{E}_{r_s,r_e}^{(n)}(\rho),
\label{avp}
\end{align}
where the sum is taken in the set $\mathcal{R}_s\times\mathcal{R}_e$.
If there is no confusion, we denote $\Lambda[\mathsf{C}_{n},\mathcal{F}^{n},S_n]$ by $\Lambda_{n}$.
Then, the correctness and secrecy of the code is evaluated by the entanglement fidelity 
\begin{align}
F_e^2(\rho_{\mathrm{mix}},\Lambda_{n}):=\langle \Phi |\Lambda_{n}\otimes\iota_R(\kb{\Phi})|\Phi\rangle  \label{eq:ent_fid}
\end{align}
of the completely mixed state $\rho_{\mathrm{mix}}$ on $\cH_{\mathrm{code}}^{(n)}$ and the averaged protocol $\Lambda[\mathsf{C}_{n},\mathcal{F}^{n},S_n]$,
where $|\Phi\rangle$ is the maximally entangled state and $\iota_R$ is the identity operator on the reference system.

\begin{theo}[Quantum Network Code with Negligible Rate Secret Shared Randomness] \label{theo:SQNC_protocol}
Suppose that the sender and receiver can share any secret randomness of negligible size in comparison with the block-length.
When $m_1 < m_0/2$,
{there exist a sequence $\{n_{\ell}\}_{\ell=1}^{\infty}$ with $n_{\ell} \to \infty$ as $l\to \infty$ and}
a sequence $\{\Code_{n_{\ell}}\}_{\ell=1}^{\infty}$ of quantum network codes of block-lengths $n_{\ell}$
such that 
\begin{align}
&\lim_{\ell\to\infty} \frac{|\mathcal{R}_s|}{n_{\ell}} = 0,
        \label{eq:aaaa}\\
&\lim_{\ell\to\infty} \frac{ \log_q \dim \cH_{\mathrm{code}}^{(n_{\ell})}}{n_{\ell}} = m_0-2m_1,
        \label{eq:bbbb}\\
&\lim_{\ell\to\infty} \max_{(\mathcal{F}, S_{n_{\ell}})} n_{\ell}(1-F_e^2(\rho_{\mathrm{mix}},\Lambda_{n_{\ell}}) ) = 0,
        \label{eq:ent_fid_theo}
\end{align}
where 
$\Lambda_{n_{\ell}}:= \Lambda[\mathsf{C}_{n_{\ell}},\mathcal{F}^{n_{\ell}}, S_{n_{\ell}}]$, and
{the maximum is taken 
with respect to $(\mathcal{F}, S_{n_{\ell}})$ in $\zeta_{m_0,m_1}^{(n_{\ell})}$ which is defined in \eqref{eq:set_network_operations}.}
\end{theo}

Notice that this code depends only on the rates $m_0$ and $m_1$,
and does not depend on the detailed structure $\mathcal{F}$ of the network.
Section~\ref{sec:protocol} gives the code realizing the performance mentioned in Theorem~\ref{theo:SQNC_protocol}. 
Sections~\ref{sec:analysis} and \ref{sec:errors} prove that the code in Section~\ref{sec:protocol} satisfies the performance mentioned in Theorem \ref{theo:SQNC_protocol}.
Section~\ref{sec:secrecy} shows that the condition \eqref{eq:ent_fid_theo} implies the secrecy of the code,
by using the result of \cite{Schumacher96}.


\begin{table}
\renewcommand{\arraystretch}{1.2}
    \centering
\caption{Comparison of quantum codes for $m_0$ parallel channels}   \label{tab:comparison}
\begin{tabular}[ht]{|c||c|c|c|}
\hline
    &   Quantum MDS code \cite{Rain99}    &     Our code   \\
\hline
\hline
Use of network &   one-shot    &     {asymptotically many} \\
\hline
Error probability   &   zero-error     &   vanishing error    \\
\hline
Range of $m_1$
&   $m_1 < m_0/4$   &   $m_1 < m_0/2$   \\
\hline
Rate    &   $m_0-4m_1$  &      $m_0-2m_1$    \\
\hline
\multicolumn{3}{l}{$m_0$: number of parallel channels.}\\[-0.4em]
\multicolumn{3}{l}{$m_1$: maximum number of corrupted channels.}
\end{tabular}
\end{table}

Indeed, it is known that
there exists a classical network code which transmits
classical information securely when the number of attacked channels is 
less than a half of the transmission rate from the sender to
the receiver \cite{YSJL14}.
Although Theorem~\ref{theo:SQNC_protocol} requires secure transmission of classical information with negligible rate in order for shared randomness,
the result \cite{YSJL14} implies that
such secure transmission can be realized by using our quantum network in bit basis states with the negligible number of times.
Hence, as shown in Section~\ref{sec:SQNC}, the combination of the result \cite{YSJL14} and 
Theorem~\ref{theo:SQNC_protocol} yields the following theorem.

\begin{theo}[Quantum Network Code without Classical Communication] \label{theo:SQNC_protocol2}
When $m_1 < m_0/2$,
{there exist a sequence $\{n_{\ell}\}_{\ell=1}^{\infty}$ with $n_{\ell} \to \infty$ as $l\to \infty$ and}
a sequence $\{\Code_{n_{\ell}}\}_{\ell=1}^{\infty}$ of quantum network codes of block-lengths $n_{\ell}$
such that 
\begin{align}
&|\mathcal{R}_s| = 0,\label{CX}\\
&\lim_{\ell\to\infty} \frac{ \log_q \dim \cH_{\mathrm{code}}^{(n_{\ell})}}{n_{\ell}} = m_0-2m_1,
\label{AX}\\
&\lim_{\ell\to\infty} \max_{(\mathcal{F}, S_{n_{\ell}})} n_{\ell}(1-F_e^2(\rho_{\mathrm{mix}},\Lambda_{n_{\ell}}) ) = 0,
        \label{BX}
\end{align}
where 
$\Lambda_{n_{\ell}}:= \Lambda[\mathsf{C}_{n_{\ell}},\mathcal{F}^{n_{\ell}}, S_{n_{\ell}}]$,
and
{the maximum is taken 
with respect to $(\mathcal{F}, S_{n_{\ell}})$ in $\zeta_{m_0,m_1}^{(n_{\ell})}$ which is defined in \eqref{eq:set_network_operations}.}
\end{theo}

\subsection{Comparison our code with quantum error-correcting code and quantum secret sharing}

%

To compare with existing results, we consider the special case where the network consists of $m_0$ parallel channels.
The quantum maximum distance separable (MDS) code \cite{Rain99} of length $m_0$ works in this network even for the one-shot setting which means one use of the network.
When $m_1 < m_0/4$ and at most $m_1$ channels are corrupted, 
the code has the rate $m_0 - 4m_1$ and the error is zero.
On the other hand, 
our code works with $n$ uses of the same network, and the position of $m_1$ corrupted channels
is assumed to be fixed over all network uses. 
Then, when $m_1 < m_0/2$ and at most $m_1$ channels are corrupted, 
our code has the rate $m_0 - 2m_1$ and 
the error goes to zero as the number $n$ of network use goes to infinity.

On the other hand, our code has an advantage that it can be used in any networks defined in Section \ref{sec:model} without any modification of the code,
whereas the quantum MDS code \cite{Rain99} works only in the network with $m_0$ parallel channels.

%
%
%
Our code applied for $m_0$ parallel channels can be thought of as an error-tolerant quantum secret sharing \cite{CGS05}.
In error-tolerant quantum secret sharing, 
a sender encodes a secret to $m_0$ shares and distributes the shares to $m_0$ players,
and all players send their shares to the receiver.
If $m_0-m_1$ players are honest,
even if the other $m_1$ players send maliciously corrupted shares,
the receiver can recover the secret and the secret is not leaked to the malicious players.
Our code implements this task if the majority of players are honest, i.e., $m_1 < m_0/2$, which is the same for the error-tolerant quantum secret sharing scheme in \cite{CGS05}.

\section{Preliminaries} \label{sec:prelim}
In this section, we prepare definitions and notations which are necessary for our code construction in Section \ref{sec:protocol}.
{In the remainder of this paper, we assume $m_a \leq m_1 < m_0/2$.}

\subsection{Phase basis} \label{sec:notations}
Let $q=s^t$ for a prime number $s$ and a positive integer $t$.
In the construction of our code, we will discuss operations on the phase basis $\{ |z\rangle_p \}_{z\in\mathbb{F}_q}$ which is defined as \cite[Section 8.1.2]{Haya2}
\begin{align*}
|z\rangle_p := \frac{1}{\sqrt{q}} \sum_{x\in\mathbb{F}_q} \omega^{-\tr (xz)} |x\rangle_b
\end{align*}
for $\omega := \exp(2\pi i/s)$ and $\tr y:= \Tr M_y$ ($\forall y\in\mathbb{F}_q$).
Here, the matrix $M_y\in\mathbb{F}_s^{t\times t}$ is the multiplication matrix $x\in\mathbb{F}_q \mapsto yx \in\mathbb{F}_q$ where the finite field $\mathbb{F}_q$ is identified with the vector space $\mathbb{F}_s^t$.

The following Lemma \ref{lemm:invertible_to_unitary}
describes the application of the unitaries $\lefto{A}$ and $\righto{A}$, defined in \eqref{eq:chops}, to the phase basis states, and is proved in Appendix A.
\begin{lemm}\label{lemm:invertible_to_unitary}
For any $Z\in\mathbb{F}_q^{m\times n}$ and any invertible matrices $A \in \mathbb{F}_q^{m\times m}$ and $B \in \mathbb{F}_q^{n\times n}$,
we have
\begin{align}
\lefto{A} |Z\rangle_p = | (A^{\top})^{-1}Z \rangle_p,\quad\!\!\!
\righto{B}|Z\rangle_p = | Z(B^{\top})^{-1} \rangle_p. 
\end{align}
\end{lemm}
For convenience, we use notation $[A]_p := (A^{-1})^{\top} = (A^{\top})^{-1} $ for any invertible matrix $A$.

\subsection{Block-lengths and extended quantum system in our code} \label{subsec:extended}

First, we define the sequence $\{n_{\ell}\}_{\ell=1}^{\infty}$ of block-lengths.
For any positive integer $\ell$,
define four parameters
\begin{gather}
\alpha_{\ell} := \max\{\floor*{ 5\log_q \ell },1\},\quad
n_\ell' := \floor*{\frac{\ell}{\alpha_{\ell}}}, \nonumber \\
n_{\ell} := \alpha_{\ell} n_\ell',\quad
q' := q^{\alpha_{\ell}}. \label{def:block-length}
\end{gather}
Then, we have
\begin{gather}
\lim_{\ell\to\infty} \frac{n_{\ell}\cdot (n_{\ell}')^{m_0}}{(\Q)^{m_0-m_1}} = 0,    \label{cond:nq}
\end{gather}
because
\begin{align*}
&\frac{n_{\ell}\cdot (n_{\ell}')^{m_0}}{(\Q)^{m_0-m_1}}  \leq \frac{\ell^{1+m_0}}{q^{(5\log_q \ell-1)(m_0-m_1)}}\\
&\leq \frac{\ell^{1+5m_1-4m_0}}{q^{m_1-m_0}} \leq \frac{\ell^{1-1.5m_0}}{q^{m_1-m_0}} \to 0.
\end{align*}
In the following, 
we construct our code only for any sufficiently large $\ell$ such that the condition 
\begin{align}
n_{\ell}' \geq 3m_0  \label{cond:nprime}
\end{align}
holds, which is enough to discuss the asymptotic performance of the code.

In our code, an extended quantum system $\cHp:=\cH^{\otimes \alpha_\ell}$ is the unit quantum system for encoding and decoding operations.
We identify the system $\mathcal{H}'$
with the system spanned by $\px{|x\rangle_b \mid x \in \FF_{\Q}}$.
Then, {$n_\ell$} uses of the network over $\cH$ can be regarded as
$n_\ell'$ uses of the network over $\mathcal{H}'$.
For invertible matrices $A \in \mathbb{F}_{\Q}^{m\times m}$ and $B \in \mathbb{F}_{\Q}^{n\times n}$,
two unitaries $\leftop{A}$ and $\rightop{B}$ are defined, similarly to \eqref{eq:chops}, as
\begin{align*}
\leftop{A} &:= \!\!\!\sum_{X\in\FF_{\Q}^{m \times n}}\ket{AX}_{bb}\bra{X}, \
 \!\p\! 
 \rightop{B} :=\!\!\! \sum_{X\in\FF_{\Q}^{m \times n}}\ket{XB}_{bb}\bra{X},
\end{align*}
and similarly to Lemma \ref{lemm:invertible_to_unitary},
for any $Z\in\mathbb{F}_{q'}^{m\times n}$, we have 
\begin{align*}
\leftop{A} |Z\rangle_p = | (A^{\top})^{-1}Z \rangle_p,\quad
\rightop{B}|Z\rangle_p = | Z(B^{\top})^{-1} \rangle_p. 
\end{align*}

\subsection{Notations for quantum systems and states} \label{subsec:sys}

In this subsection, we introduce several notations for quantum states and systems.
For the quantum system $\cH^{\otimes m_0\times n_\ell}= (\cHp)^{\otimes m_0\times n_\ell'}$ which is transmitted by $n_\ell$ uses of the network,
we use the following notation:
\begin{align*}
                        &(\cHp)^{\otimes m_0\times n_\ell'} \!\!= \cHp_\mathcal{A}\otimes \cHp_\mathcal{B} \otimes \cHp_\mathcal{C}\\
                        &:= (\cHp)^{\otimes m_0\!\times\! \mm}\!\otimes\! (\cHp)^{\otimes m_0\!\times\! \mm} \!\otimes\! (\cHp)^{\otimes m_0\!\times \!(n_\ell'\!-\!2\mm\!)}.
\end{align*}
Moreover, for any $\mathcal{X}\in\{\mathcal{A},\mathcal{B},\mathcal{C}\}$ and $(m_\mathcal{A}, m_\mathcal{B},m_\mathcal{C}):= (m_0,m_0,n_\ell'-2m_0)$, we denote
\begin{align*}
    &\cHp_\mathcal{X} = \cHp_{\mathcal{X}1}\otimes \cHp_{\mathcal{X}2}\otimes \cHp_{\mathcal{X}3}   \\
       &:= (\cHp)^{\otimes m_1\times m_\mathcal{X}}\otimes (\cHp)^{\otimes (m_0-2m_1)\times m_\mathcal{X}}\otimes (\cHp)^{\otimes m_1\times m_\mathcal{X}}.
\end{align*}
The tensor product state of $|\phi\rangle\in \cHp_{\mathcal{X}1}, 
                |\psi\rangle\in \cHp_{\mathcal{X}2}$, and 
                $|\varphi\rangle\in \cHp_{\mathcal{X}3}$ is denoted as
\begin{align*}
    \begin{bmatrix}
    |\phi\rangle \\ 
    |\psi\rangle \\ 
    |\varphi\rangle \\
    \end{bmatrix}
    :=
    |\phi\rangle 
    \otimes |\psi\rangle 
    \otimes |\varphi\rangle
    \in
    \cHp_\mathcal{X}.
\end{align*}
For any {block matrix}
$[X^{\top}, Y^{\top}, Z^{\top}]^{\top} \in \mathbb{F}_q^{m_1\times m_\mathcal{X}}\times \mathbb{F}_q^{(m_0-2m_1)\times m_\mathcal{X}}\times \mathbb{F}_q^{m_1\times m_\mathcal{X}}$,
the bit and phase basis states of $[X^{\top}, Y^{\top}, Z^{\top}]^{\top}$ are denoted by
\begin{align*}
    \ket*{
    \begin{bmatrix}
    X \\ 
    Y \\ 
    Z \\
    \end{bmatrix}}_{\!\!b}
    := 
    \begin{bmatrix}
    \ket{X}_b \\ 
    |Y\rangle_b \\ 
    |Z\rangle_b \\
    \end{bmatrix}
    ,
    \quad
    \ket*{
    \begin{bmatrix}
    X \\ 
    Y \\ 
    Z \\
    \end{bmatrix}
    }_{\!\!p}
    := 
    \begin{bmatrix}
    |X\rangle_p \\ 
    |Y\rangle_p \\ 
    |Z\rangle_p \\
    \end{bmatrix}
    .
\end{align*}
{The $k\times l$ zero matrix is denoted by $\zerom{k}{l}$, and $|i,j\rangle := |i\rangle \otimes |j\rangle$.}

\subsection{CSS code in our quantum network code} \label{sec:css_code}

In this subsection, we define a {Calderbank--Steane--Shor (CSS) code \cite{Steane96,CS96,Steane96-2}} which is used in the construction of our quantum network code in Section~\ref{sec:protocol}.
A CSS code is defined from two classical codes $C_1$ and $C_2$ satisfying $C_1 \supset C_2^{\perp}$, where a classical code is defined as the set of codewords.
Therefore, in order to define the CSS code used in our code, we define the following two classical codes:
{by identifying the set $\mathbb{F}_{\Q}^{m_0\times(n_\ell'\!-\!2\mm\!)}$ of matrices with the vector space $\mathbb{F}_{\Q}^{m_0(n_\ell'\!-\!2\mm\!)}$,}
the classical codes $C_1, C_2\subset \mathbb{F}_{\Q}^{m_0\times (n_\ell'\!-\!2\mm\!)}$ are defined by 
\begin{align*}
C_1 :=
    \mleft\{
     \!\!   \vphantom{     \begin{bmatrix}
                        \zerom{m_1}{n_\ell'\!-\!2\mm\!} \\ 
                        Y \\ 
                        Z \\
                        \end{bmatrix}
                        }
    \mright.
    &
    \mleft.
    \begin{bmatrix}
    \zerom{m_1}{n_\ell'\!-\!2\mm\!} \\ 
    Y \\ 
    Z \\
    \end{bmatrix}
    \in \mathbb{F}_{\Q}^{m_0\times(n_\ell'\!-\!2\mm\!)}
    \;\middle|\;
    \mright.
    \\
    &\qquad
    \mleft. \!\!\!\vphantom{     \begin{bmatrix}
                        \zerom{m_1}{n_\ell'\!-\!2\mm\!} \\ 
                        Y \\ 
                        Z \\
                        \end{bmatrix}
                        }
        Y\in \mathbb{F}_{\Q}^{(m_0-2m_1)\times(n_\ell'\!-\!2\mm\!)},\
        Z\in \mathbb{F}_{\Q}^{m_1\times(n_\ell'\!-\!2\mm\!)}
    \mright\},\\
C_2 := 
    \mleft\{
      \!\!  \vphantom{     \begin{bmatrix}
                        \zerom{m_1}{n_\ell'\!-\!2\mm\!} \\ 
                        Y \\ 
                        Z \\
                        \end{bmatrix}
                        }
    \mright.
    &\mleft.
    \begin{bmatrix}
    X \\ 
    Y \\ 
    \zerom{m_1}{n_\ell'\!-\!2\mm\!} \\
    \end{bmatrix}
    \in \mathbb{F}_{\Q}^{m_0\times(n_\ell'\!-\!2\mm\!)}
    \;\middle|\;
    \mright.
    \\
    &\qquad
    \mleft. \!\!\!\vphantom{     \begin{bmatrix}
                        \zerom{m_1}{n_\ell'\!-\!2\mm\!} \\ 
                        Y \\ 
                        Z \\
                        \end{bmatrix}
                        }
        X \in \mathbb{F}_{\Q}^{m_1\times(n_\ell'\!-\!2\mm\!)},\ 
        Y \in \mathbb{F}_{\Q}^{(m_0-2m_1)\times(n_\ell'\!-\!2\mm\!)}
    \mright\}.
\end{align*}
The classical codes $C_1$ and $C_2$ satisfy $C_1 \supset C_2^{\perp} = \{ [\zerom{m_1}{n_\ell'\!-\!2m_0\!}^{\top}, \zerom{m_0-2m_1}{n_\ell'\!-\!2m_0\!}^{\top}, Z^{\top}]^{\top} \mid Z\in \mathbb{F}_{q'}^{m_1\times(n_\ell'\!-\!2m_0\!)}\}$.
For any coset $M+C_2^{\perp} \in C_1/C_2^{\perp}$ 
containing $M\in\mathbb{F}_{\Q}^{ (m_0\!-\!2m_1)\times (n_\ell'\!-\!2\mm)}$,
define a quantum state $|M+C_2^{\perp}\rangle_b \in \cHp_\mathcal{C}$ by
\begin{align}
    |M+C_2^{\perp}\rangle_b &:=
    \frac{1}{\sqrt{|C_2^{\perp}|}} 
    \sum_{J\in C_2^{\perp}}
    \left|
    \begin{bmatrix}
    \zerom{m_1}{n_\ell'\!-\!2\mm\!}\\ 
    M \\ 
    \zerom{m_1}{n_\ell'\!-\!2\mm\!} \\
    \end{bmatrix}
    +J \right\rangle_{\!\!b}  \nonumber\\
    &=
    \begin{bmatrix}
    |\zerom{m_1}{n_\ell'\!-\!2\mm\!} \rangle_b \\ 
    |M\rangle_b \\ 
    |\zerom{m_1}{n_\ell'\!-\!2\mm\!} \rangle_p \\
    \end{bmatrix}. \nonumber
\end{align}
Then, the CSS code is defined as 
$\mathsf{CSS}(C_1,C_2):=\{ |M+C_2^{\perp}\rangle_b \mid M\in\mathbb{F}_{\Q}^{ (m_0\!-\!2m_1)\times (n_\ell'\!-\!2\mm)}\}$.
That is, 
any state $|\phi\rangle 
\in \cH_{\mathrm{code}}^{(n_\ell)} := \cHp_{\mathcal{C}2} = (\cHp)^{\otimes (m_0-2m_1)\times (n_\ell'\!-\!2\mm\!)}$
is encoded as
\begin{align}
    \begin{bmatrix}
    |\zerom{m_1}{n_\ell'\!-\!2\mm\!} \rangle_b \\ 
    |\phi\rangle \\ 
    |\zerom{m_1}{n_\ell'\!-\!2\mm\!} \rangle_p \\
    \end{bmatrix}
    \in 
    {\text{span}\ \!\mathsf{CSS}(C_1,C_2)} 
    \subset \cHp_\mathcal{C}. \nonumber
\end{align}
The above CSS code is used in our code construction.

{
\subsection{Other Notations}    \label{subsec:other_notations}

In correspondence with the notations in Section \ref{subsec:sys},
for any positive integer $k$ and any matrix $X\in\mathbb{F}_{\Q}^{k \times n_\ell'}$, we denote 
\begin{align*}
X = [X^\mathcal{A},X^\mathcal{B},X^\mathcal{C}] \in \mathbb{F}_{\Q}^{k\times \mm}\times \mathbb{F}_{\Q}^{k\times \mm}\times \mathbb{F}_{\Q}^{k\times (n_\ell'\!-\!2\mm\!)}.
\end{align*}
{If $k=m_0$, 
for any $\mathcal{X}\in\{\mathcal{A},\mathcal{B},\mathcal{C}\}$,
we denote
$X^{\mathcal{X}} = [ (X^{\mathcal{X}1})^{\top}, (X^{\mathcal{X}2})^{\top}, (X^{\mathcal{X}3})^{\top}]^{\top}$,
where $X^{\mathcal{X}1},X^{\mathcal{X}3}\in \mathbb{F}_{\Q}^{m_1\times m_{0}}$ and $X^{\mathcal{X}2}\in \mathbb{F}_{\Q}^{(m_0-2m_1)\times (n_{\ell}'-2m_0})$.}

{$\pr_{R}[A(R)]$ denotes the probability that the random variable $R$ satisfies the condition $A$,}
and
$\pr_R[A(R)|B(R)]$ denotes the conditional probability that the variable $R$ satisfies the condition $A$ under the condition $B$.
}

\section{Code Construction with Negligible Rate Secret Shared Randomness} \label{sec:protocol}

Now, we describe our quantum network code with the secret shared randomness of negligible rate by $n_\ell$ network uses.

In our code, the encoder and decoder are determined depending on secret randomnesses. 
Let $\mathcal{R}_e$ be the set of $m_0\times m_0$ invertible matrices over $\mathbb{F}_{q'}$,
$\mathcal{R}_1$ be the finite field $\mathbb{F}_{q'}$,
and $\mathcal{R}_2$ be the set of $(m_0-m_1)\times m_0$ matrices over $\mathbb{F}_{q'}$ of rank $m_0-m_1$.
The private randomness $R_e$ of the encoder is uniformly chosen from $\mathcal{R}_e$.
The secret shared randomness $R_s := (\vsec,R_2):= ((\vsec_1,\ldots, \vsec_{4m_0}), (R_{2,b},R_{2,p}))$ between the encoder and decoder is uniformly chosen from $\mathcal{R}_s := \mathcal{R}_1^{4m_0} \times \mathcal{R}_2^{2}$.
Note that the size of the shared secret randomness $R_s$ is 
{less than $\log_q |\mathbb{F}_{\Q}^{4m_0}\times\mathbb{F}_{\Q}^{2(m_0-m_1)\times m_0}| = \alpha_{\ell}(2m_0^2+(4-2m_1)m_0)$ 
and therefore} negligible with respect to $n_\ell$.

The code space is $\cH_{\mathrm{code}}^{(n_\ell)}:=\cHp_{\mathcal{C}2} = (\cHp)^{\otimes (m_0-2m_1)\times (n_\ell'\!-\!2\mm\!)}$ which is the code space of the CSS code defined in Section \ref{sec:css_code}.
The encoder $\mathsf{E}^{(n_\ell)}_{R_e,R_s}$ is defined depending on $R_e$ and $R_s$
as an isometry quantum channel from $\cH_{\mathrm{code}}^{(n_\ell)}$ to $\cH^{\otimes m_0\times n_\ell}$,
and
the decoder $\mathsf{D}^{(n_\ell)}_{R_s}$ is defined depending on $R_s$
as a TP-CP map from $\cH^{\otimes m_0\times n_\ell}$ to $\cH_{\mathrm{code}}^{(n_\ell)}$.
In the following subsections,
we give the details of the encoder $\mathsf{E}^{(n_\ell)}_{R_e,R_s}$ and the decoder $\mathsf{D}^{(n_\ell)}_{R_s}$.

\subsection{Encoder $\mathsf{E}^{(n_\ell)}_{R_e,R_s}$}
For any input state $|\phi\rangle\in\cH_{\mathrm{code}}^{(n_\ell)}$,
the encoder $\mathsf{E}^{(n_\ell)}_{R_e,R_s}$ is described as follows.

\vspace{-0.5em}
\FrameSep5pt
\begin{framed}
\noindent{\bf Encode 1 (Check Bit Embedding)\quad}
    Encode the input state $|\phi\rangle$ by an isometry map $U_1^{R_2}: \cH_{\mathrm{code}}^{(n_\ell)}\to
    (\cH')^{\otimes m_0\times n_\ell'}  = \cHp_\mathcal{A}\otimes \cHp_\mathcal{B} \otimes \cHp_\mathcal{C}$ which is defined as
\begin{align*}
|\phi_1\rangle \!&:= U_1^{R_2} |\phi\rangle     \\
&=
    \bitk*{
    \left[
    \!\!\!
    \begin{array}{c}
    \zerom{m_1}{\mm} \\
    \multicolumn{1}{c}{\multirow{2}{*}{$R_{2,b}$}}\\
    \multicolumn{1}{c}{} 
    \end{array}
    \!\!\!
    \right]}
    \!\!\otimes
    \phasek*{
    \left[
    \!\!\!
    \begin{array}{c}
    \multicolumn{1}{c}{\multirow{2}{*}{${R_{2,p}}$}} \\
    \multicolumn{1}{c}{} \\
    \zerom{m_1}{\mm} \\
    \end{array}
    \!\!\!
    \right]}
    \!\!\otimes\!\!
    \left[
    \!\!\!
    \begin{array}{c}
    |\zerom{m_1}{n_\ell'\!-\!2\mm\!}\rangle_{\!b}\\
    |\phi\rangle \\
    |\zerom{m_1}{n_\ell'\!-\!2\mm\!}\rangle_{\!p} \\
    \end{array}
    \!\!\!
    \right]\!\!.
\end{align*}

\noindent{\bf Encode 2 (Vertical Mixing)\quad}
{Encode $|\phi_1\rangle$ as} 
    \begin{align*}
    |\phi_2\rangle &:= \leftop{R_e} |\phi_1\rangle \in(\cH')^{\otimes m_0\times n_\ell'} .
    \end{align*}

\noindent{\bf Encode 3 (Horizontal Mixing)\quad}
{From the shared randomness $\vsec$,} define matrices $Q_{1;i,j}:=(\vsec_j)^i$, $Q_{2;i,j}:=(\vsec_{\mm+j})^i$ for $1\leq i \leq n_\ell'\!-\!2\mm$, $1\leq j\leq \mm$,
and $Q_{3;i,j}:=(\vsec_{2\mm+j})^i$, $Q_{4;i,j}:=(\vsec_{3\mm+j})^i$ for $1\leq i \leq \mm$ and $1\leq j\leq \mm$.
With these matrices, define a random matrix $R_1^{\vsec}\in\mathbb{F}_{\Q}^{n_\ell'\times n_\ell'}$ as
\begin{align*}
R_1^{\vsec}:=&
\left[
\!\!
\begin{array}{ccc}
I_{\mm} & \zerom{\mm}{\mm} & \zerom{\mm}{n_\ell'-2\mm} \\
Q_3^{\top}+Q_4 & I_{\mm} & \zerom{\mm}{n_\ell'-2\mm} \\
\zerom{n_\ell'-2\mm}{\mm} & \zerom{n_\ell'-2\mm}{\mm} & I_{n_\ell'-2\mm}
\end{array}
\!\!
\right]
\\
 &\cdot
\left[
\!\!
\begin{array}{ccc}
I_{\mm} & \zerom{\mm}{\mm} & \zerom{\mm}{n_\ell'-2\mm}  \\
\zerom{\mm}{\mm} & I_{\mm} & Q_2^{\top} \\
\zerom{n_\ell'-2\mm}{\mm} & \zerom{n_\ell'-2\mm}{\mm} & I_{n_\ell'-2\mm}
\end{array}
\!\!
\right]
\\
 &\cdot
\left[\!\!
\begin{array}{ccc}
I_{\mm} & \zerom{\mm}{\mm} & \zerom{\mm}{n_\ell'-2\mm}  \\
\zerom{\mm}{\mm} & I_{\mm} & \zerom{\mm}{n_\ell'-2\mm} \\
Q_1 & \zerom{n_\ell'-2\mm}{\mm} & I_{n_\ell'-2\mm}
\end{array}
\!\!\right]
,
\end{align*}
where $I_d$ is the $d$-dimensional identity matrix.

{Encode $|\phi_2\rangle$ as}
    \begin{align*}
    |\phi_3\rangle &:= \rightop{R_1^{\vsec}} |\phi_2\rangle \in (\cH')^{\otimes m_0\times n_\ell'}.
    \end{align*}
\end{framed}
By the above three steps, the encoder  $\mathsf{E}^{(n_\ell)}_{R_e,R_s}$ is written as the isometry map
\begin{align*}
\mathsf{E}^{(n_\ell)}_{R_e,R_s}: |\phi\rangle\mapsto 
\rightop{R_1^{\vsec}}\leftop{R_e}U_1^{R_2} |\phi\rangle  \in \cH^{\otimes m_0\times n_\ell}.
\end{align*}

\subsection{Decoder $\mathsf{D}^{(n_\ell)}_{R_s}$} 
For any input state $|\psi\rangle  \in (\cH')^{\otimes m_0\times n_\ell'} = \cH^{\otimes m_0\times n_\ell}$,
the decoder $\mathsf{D}^{(n_\ell)}_{R_s}$ is described as follows.

\begin{framed}
\noindent{\bf Decode 1 (Decoding of Encode 3)\quad}
The inverse of $R_1^{\vsec}$ is derived from the shared randomness $\vsec$ as
\begin{align*}
\!\!
(R_1^{\vsec})^{-1}
\!
:=&
\!
\left[\!\!
\begin{array}{ccc}
I_{\mm} & \zerom{\mm}{\mm} & \zerom{\mm}{n_\ell'-2\mm} \\
\zerom{\mm}{\mm} & I_{\mm} & \zerom{\mm}{n_\ell'-2\mm} \\
-Q_1 & \zerom{n_\ell'-2\mm}{\mm} & I_{n_\ell'-2\mm}
\end{array}
\!\!\right]
\\
&\cdot
\left[\!\!
\begin{array}{ccc}
I_{\mm} & \zerom{\mm}{\mm} & \zerom{\mm}{n_\ell'-2\mm} \\
\zerom{\mm}{\mm}  & I_{\mm} & -Q_2^{\top} \\
\zerom{n_\ell'-2\mm}{\mm} & \zerom{n_\ell'-2\mm}{\mm} & I_{n_\ell'-2\mm}
\end{array}
\!\!\right]
\\
&\cdot
\left[\!\!
\begin{array}{ccc}
I_{\mm} & \zerom{\mm}{\mm} & \zerom{\mm}{n_\ell'-2\mm} \\
-Q_3^{\top}-Q_4 &  I_{\mm} & \zerom{\mm}{n_\ell'-2\mm} \\
\zerom{n_\ell'-2\mm}{\mm} & \zerom{n_\ell'-2\mm}{\mm} & I_{n_\ell'-2\mm}
\end{array}
\!\!\right]
\!.
\end{align*}
Apply $\rightop{R_1^{\vsec}}^{\dagger} = \rightop{(R_1^{\vsec})^{-1}}$ to the state $|\psi\rangle$:
\begin{align*}
 |\psi_1\rangle 
 &:= \rightop{R_1^{\vsec}}^{\dagger} | \psi \rangle 
 \in (\cH')^{\otimes m_0\times n_\ell'}  = \cHp_\mathcal{A}\otimes \cHp_\mathcal{B} \otimes \cHp_\mathcal{C}.
\end{align*}
\noindent{\bf Decode 2 (Error Correction)\quad}
{Perform the bit basis measurement $\px{ |O_b\rangle_b \mid  O_b\in\mathbb{F}_{\Q}^{m_0\times m_0}}$ on $\cHp_\mathcal{A}$ and the phase basis measurement $\px{ |O_p\rangle_p \mid  O_p\in\mathbb{F}_{\Q}^{m_0\times m_0}}$ on $\cHp_\mathcal{B}$.
The bit and phase measurement outcomes are denoted as $O_b, O_p\in\mathbb{F}_{\Q}^{m_0 \times \mm}$, respectively.}

Next, find invertible matrices $D_b, D_p \in \mathbb{F}_{\Q}^{m_0\times m_0}$ which satisfy
\begin{align}
P_{b} D_b O_b  &=  \begin{bmatrix} \zerom{m_1}{\mm} \\ R_{2,b} \end{bmatrix} , \label{eq:Gauss_bit} \\
P_{p} [D_p]_p O_p  &=  \begin{bmatrix} R_{2,p} \\ \zerom{m_1}{\mm} \end{bmatrix} , \label{eq:Gauss_phase}
\end{align}
where $P_{b}$ is the projection to the last $m_0-m_1$ elements in $\mathbb{F}_{q'}^{m_0}$ and 
 $P_{p}$ is the projection to the first $m_0-m_1$ elements in $\mathbb{F}_{q'}^{m_0}$.
If the invertible matrix $D_b$ or $D_p$ does not exist, the decoder applies no operation and returns the transmission failure.
If $D_b$ or $D_p$ is not unique, the decoder decides $D_b$ or $D_p$ deterministically depending on $O_b,R_{2,b},O_p,R_{2,p}$.
    

    Finally, apply $\leftop{D_b}$ and $\leftop{D_p}$ to the system $\cHp_\mathcal{C}$,
    and 
    output the reduced state on $\cHp_{\mathcal{C}2} = \cH_{\mathrm{code}}^{(n_\ell)}$. 
  
    Decode 2 is summarized as a TP-CP map $\mathsf{D}_2$ from $\cHp_\mathcal{A}\otimes\cHp_\mathcal{B}\otimes\cHp_\mathcal{C}$ to $\cH_{\mathrm{code}}^{(n_\ell)}$ by
    \begin{align*}
    &\mathsf{D}_2(\kb{\psi_1})\\
    &:=
    \Tr_{\mathcal{C}1,\mathcal{C}3}
    \!\!\!\!\!\!
    \sum_{O_{\!b},O_{\!p}\in\mathbb{F}_{\Q}^{m_0\times \mm}}
    \!\!\!\!\!\!
    \mathbf{D}_3^{R_2\!,O_{\!b}\!,O_{\!p}}
    \rho_{O_{\!b}\!,O_{\!p}\!,|\psi_1\rangle}
    (\mathbf{D}_3^{R_2\!,O_{\!b}\!,O_{\!p}})^\dagger
    ,
    \end{align*}
    where the matrix $\rho_{O_{\!b}\!,O_{\!p}\!,|\psi_1\rangle}$ and the unitary $\mathbf{D}_3^{R_2\!,O_{\!b}\!,O_{\!p}}$ are defined as
    \begin{align*}
    \rho_{O_{\!b}\!,O_{\!p}\!,|\psi_1\rangle}
    &:=
    \Tr_{\mathcal{A},\mathcal{B}} 
    \kb{\psi_1}
    (|O_b\rangle_{b}{}_b\langle O_{\!b}|
    \otimes
    |O_p\rangle_{pp}\langle O_{\!p}|
    \otimes I_\mathcal{C}),\\
    \mathbf{D}_3^{R_2\!,O_{\!b}\!,O_{\!p}}
    &:=\leftop{D_p}
      \leftop{D_b}.
    \end{align*}
\end{framed}
By the above two steps, the decoder $\mathsf{D}^{(n_\ell)}_{R_s}$ is written as the TP-CP map
\begin{align*}
\mathsf{D}^{(n_\ell)}_{R_s}(\kb{\psi}) = \mathsf{D}_2 \paren*{\rightop{R_1^{\vsec}}^{\dagger} \kb{\psi}\rightop{R_1^{\vsec}} }.
\end{align*}

The performance of our code will be analyzed in Section~\ref{sec:analysis}.

\section{Analysis of Our Code} \label{sec:analysis}

In this section, we evaluate the performance of the code in Section \ref{sec:protocol}.
That is, we show that the code in Section \ref{sec:protocol} satisfies the conditions \eqref{eq:aaaa}, \eqref{eq:bbbb}, and \eqref{eq:ent_fid_theo} in Theorem \ref{theo:SQNC_protocol}.

First, we evaluate the size of the secret shared randomness and the rate of the code. 
The size of the secret shared randomness $R_s$ is less than $\log_q |\mathbb{F}_{\Q}^{4m_0}\times\mathbb{F}_{\Q}^{2(m_0-m_1)\times m_0}| = \alpha_{\ell}(2m_0^2+(4-2m_1)m_0)$ which does not scale with the block-length $n_\ell$.
Therefore, the secret shared randomness is negligible, i.e., the condition \eqref{eq:aaaa} is satisfied.
Moreover, since the dimension of the code space $\cH_{\mathrm{code}}^{(n_\ell)}$ is $(\Q)^{(m_0-2m_1)(n_\ell'-2m_0)}=q^{(m_0-2m_1)(n_\ell-2m_0\alpha_{\ell} )}$,
the rate of our code is $m_0-2m_1$, i.e., the condition \eqref{eq:bbbb} is satisfied.


Next, we evaluate the correctability of the code.
That is, we show that our code satisfies the condition \eqref{eq:ent_fid_theo}, i.e.,
\begin{align*}
\lim_{\ell\to\infty} \max_{(\mathcal{F},S_{n_\ell})} n_\ell(1-F_e^2(\rho_{\mathrm{mix}},\Lambda_{n_\ell}) ) = 0.
\end{align*}

Recall that the averaged protocol is written in \eqref{avp} as 
\begin{align*}
&\Lambda_{n_\ell}= \Lambda[\mathsf{C}_{{n_\ell}},\mathcal{F}^{{n_\ell}},S_{n_\ell}](\rho)\\
&=  \frac{1}{|\mathcal{R}_s\times\mathcal{R}_e|} \!\sum_{(r_s,r_e)\in\mathcal{R}_s\times\mathcal{R}_e}\!\!\!\!\!\!\! \mathsf{D}_{r_s}^{(n_\ell)} \circ \Gamma[\mathcal{F}^{{n_\ell}},S_{n_\ell}] \circ \mathsf{E}_{r_s,r_e}^{(n_\ell)}(\rho),
\end{align*}
and the entanglement fidelity is written in \eqref{eq:ent_fid} as
\begin{align*}
F_e^2(\rho_{\mathrm{mix}},\Lambda_{n_\ell})=\langle \Phi |\Lambda_{n_\ell}\otimes\iota_R(\kb{\Phi})|\Phi\rangle.
\end{align*}
Here, the maximally entangled state $|\Phi\rangle$ is written as $|\Phi\rangle := \big( 1/{(\Q)}^{m/2}\big)\sum_{x\in\mathbb{F}_{\Q}^{m}} |x, x\rangle_b$ for $m := (m_0-2m_1)(n_\ell'\!-\!2\mm\!)$
since $\cH_{\mathrm{code}}^{(n_\ell)} = (\cH')^{m}$.
The entanglement fidelity is evaluated by 
\begin{align}
 &1-F_e^2(\rho_{\mathrm{mix}}, \Lambda_{n_\ell})\\
 =& 1-\langle \Phi |\Lambda_{n_\ell}\otimes\iota_R(\kb{\Phi})|\Phi\rangle\nonumber\\
 =& \Tr \Lambda_{n_\ell}\otimes\iota_R(\kb{\Phi})(I-P_1P_2) \label{eq:binary_PbPp}\\
 \leq& \Tr \Lambda_{n_\ell}\!\!\otimes\!\iota_{\!R}(\kb{\Phi})(\!I\!\!-\!\!P_1\!) 
    \!+\!\Tr \Lambda_{n_\ell}\!\!\otimes\!\iota_{\!R}(\kb{\Phi})(\!I\!\!-\!\!P_2\!) \!\label{ineq:bit_and_phase_ef}\!
\end{align}
for ${P_1} := \sum_{x\in\mathbb{F}_{\Q}^{m}} |x,x\rangle_{bb}\langle x,x|$ and
$P_2 := \sum_{z\in\mathbb{F}_{\Q}^{m}} |z,\bar{z}\rangle_{pp}\langle z,\bar{z}|$ where $|\bar{z}\rangle_p$ is the {complex conjugate} of $|z\rangle_p$.
The equality of (\ref{eq:binary_PbPp}) holds from $P_1P_2 = \kb{\Phi}$ which is proved in Lemma~\ref{lemm:PBPP}.

The two terms in \eqref{ineq:bit_and_phase_ef} are error probabilities with respect to the bit and phase bases, respectively, in the following sense.
Define {\em the bit error probability of $\Lambda_{n_\ell}$} as the average probability that a bit basis state $|x\rangle_b\in\cH_{\mathrm{code}}^{(n_\ell)}$ is the input state of $\Lambda_{n_\ell}$ but the bit basis measurement outcome on the output state is not $x$.
Since the bit error probability is evaluated as 
\begin{align*}
&\textnormal{(bit error probability)} \\
&= 1 - {\frac{1}{{{(\Q)}^{m}}}}\sum_{x\in\mathbb{F}_{\Q}^{m}} {}_{b}\!\langle x |\Lambda_{n_\ell}\paren*{|x\rangle_{bb}\langle x|}|x\rangle_{b}\\
&= 1 - {\frac{1}{{{(\Q)}^{m}}}}\sum_{x\in\mathbb{F}_{\Q}^{m}} \Tr P_{1}\! \cdot\! (\Lambda_{n_\ell}\!\!\otimes\!\iota_{\!R}(|x\!,\!x\rangle_{bb}\langle x\!,\!x|))\\
&= \Tr \Lambda_{n_\ell}\otimes\iota_R(\kb{\Phi})(I-P_1),
\end{align*}
the bit error probability is equal to the first term of \eqref{ineq:bit_and_phase_ef}.
Similarly, 
the second term $\Tr \Lambda_{n_\ell}\otimes\iota_R(\kb{\Phi})(I-P_2)$ of \eqref{ineq:bit_and_phase_ef} is {\em the phase error probability of $\Lambda_{n_\ell}$} which is the average probability that a phase basis state is the input of $\Lambda_{n_\ell}$ but the phase basis measurement outcome on output is incorrect.
Therefore, we can bound the entanglement fidelity as
\begin{align}
&1-F_e^2(\rho_{\mathrm{mix}}, \Lambda_{n_\ell})    \nonumber \\
&\leq \text{(bit error probability)}  + \text{(phase error probability)} .
\label{eq:entfid_twoerr}
\end{align}

The bit and phase error probabilities of our code are evaluated by the following lemma, which is proved in Section~\ref{sec:errors}.
\begin{lemm} \label{lemm:errorp}
Let $\Code_{n}$ be the quantum network code constructed in Section \ref{sec:protocol}
and 
suppose that the randomness $R_s$ of $\Code_{n}$ is shared secretly between the encoder and decoder.
{For any $(\mathcal{F},S_{n_\ell})\in\zeta_{m_0,m_1}^{(n_\ell)}$ defined in \eqref{eq:set_network_operations}},
the bit and phase error probabilities of $\Lambda[\mathsf{C}_{n_\ell},\mathcal{F}^{n_\ell},S_{n_\ell}]$ are evaluated as 
\begin{align}
\textnormal{(bit error probability)} \!\leq\! O\Big(\! \max \Big\{\frac{1}{\Q} , \frac{(n_\ell')^{m_0}}{(\Q)^{m_0-m_1}} \Big\} \!\Big), \label{biterr}\\
\!\textnormal{(phase error probability)} \!\leq\! O\Big(\! \max \Big\{\frac{1}{\Q} , \frac{(n_\ell')^{m_0}}{(\Q)^{m_0-m_1}} \Big\} \!\Big). \label{phaseerr}
\end{align}
\end{lemm}

By combining Eq. \eqref{eq:entfid_twoerr} and Lemma \ref{lemm:errorp}, we have the following inequality:
\begin{align*}
\max_{(\mathcal{F},S_{n_\ell})} 
1-F_e^2(\rho_{\mathrm{mix}}, \Lambda_{n_\ell}  )
&\leq O\paren*{ \max \px*{\frac{1}{\Q} , \frac{(n_\ell')^{m_0}}{(\Q)^{m_0-m_1}} } }. 
\end{align*}
From the condition \eqref{cond:nq},
and since the condition \eqref{cond:nq} implies $\lim_{\ell\to\infty} n_\ell/\Q= 0$,
the condition \eqref{eq:ent_fid_theo} is satisfied.

To summarize, the code in Section \ref{sec:protocol} satisfies the conditions \eqref{eq:aaaa}, \eqref{eq:bbbb}, and \eqref{eq:ent_fid_theo} in Theorem \ref{theo:SQNC_protocol}. Thus, Theorem \ref{theo:SQNC_protocol} is proved.

\section{Bit and Phase Error Probabilities} \label{sec:errors}
In this section, 
we prove Lemma \ref{lemm:errorp}, that is,
we bound separately the bit and phase error probabilities of $\Lambda_{n_\ell}$. 

\subsection{Lemmas for derivation of bit and phase error probabilities}

Before we prove Lemma \ref{lemm:errorp}, we prepare three lemmas.
The first lemma is a variant of \cite[Lemma 5]{HOKC17}. 
\begin{lemm}          \label{lemm:funda}
Let $\mathcal{V}$ be a vector space,
and $\mathcal{W}_1$ and $\mathcal{W}_2$ be subspaces of $\mathcal{V}$.
Suppose the following two conditions \textnormal{(A)} and \textnormal{(B)} hold.
\begin{itemize}
\item[(A)] $\mathcal{W}_1 \cap \mathcal{W}_2 = \{0\}$.
\item[(B)]
$n_0$ vectors $u_1+v_1, \ldots, u_{n_0}+v_{n_0} \in \mathcal{W}_1 \oplus \mathcal{W}_2$
span the subspace $\mathcal{W}_1 \oplus \mathcal{W}_2$. 
\end{itemize}

\noindent Then, the following two statements {hold.}
\begin{itemize}
\item[(C)] 
Let $\mathcal{W}_3$ be a subspace of $\mathcal{V}$ such that $\dim\mathcal{W}_3=\dim\mathcal{W}_1$.
For any bijective linear map $A$ from  $\mathcal{W}_1$ to $\mathcal{W}_3$,
there exists an invertible matrix $D$ on $\mathcal{V}$ such that
\begin{align}
P_{\mathcal{W}_3} D(u_i+v_i) = Au_i  \quad (\forall i\in\{1,\ldots,n_0\}),  \label{eq:cc}
\end{align}
where $P_{\mathcal{W}_3}$ is the projection to the subspace $\mathcal{W}_3$.
\item[(D)]
For any $u+v \in  \mathcal{W}_1 \oplus \mathcal{W}_2$, any matrix $D$ satisfying \eqref{eq:cc} satisfies
\begin{align}
P_{\mathcal{W}_3} D(u+v) = Au. \label{eq:dd}
\end{align}
\end{itemize}
\begin{proof}
From the condition \textnormal{(A)}, there exists an invertible matrix $D$ on $\mathcal{V}$ such that $Du=Au\in\mathcal{W}_3$ and $Dv\in\mathcal{W}_3^{\perp}$ for any $u\in\mathcal{W}_1$ and $v\in\mathcal{W}_2$. Then, the map $D$ satisfies \eqref{eq:cc}, which implies the condition \textnormal{(C)}.
Moreover, the condition \textnormal{(B)} guarantees that the condition \textnormal{(C)} implies the condition~\textnormal{(D)}.
\end{proof}
\end{lemm}

In addition, we also prepare the following two lemmas.

\begin{lemm} \label{lemm:space}
For any positive integers $n_0\geq n_1+n_2$,
fix an $n_0$-dimensional vector space $\mathcal{V}$ over $\mathbb{F}_q$ and 
an $n_1$-dimensional subspace $\mathcal{W}\subset \mathcal{V}$, 
and 
{let $\mathfrak{R}$ be the set of $n_2$-dimensional subspaces of $\mathcal{V}$. 
When the choice of $\mathcal{R}\in\mathfrak{R}$ follows the uniform distribution,
we have}
\begin{align*}
\pr[\mathcal{W}\cap \mathcal{R}=\{0\}] = 1-O(q^{n_1+n_2-n_0-1}),
\end{align*}
where the big-O notation is with respect to the prime power $q$ which goes to infinity. 
\begin{proof}
The probability $\pr[\mathcal{W}\cap \mathcal{R}=\{0\}]$
is the same as the probability 
to choose $n_2$ linearly independent vectors 
so that they do not intersect with $\mathcal{W}$,
{which is done by the following method:
choose $v_1$ from $\mathcal{V}\setminus \mathcal{W}$, and 
for each $i\in\{1,\ldots,n_2-1\}$, choose $v_{i+1}$ from  $\mathcal{V}\setminus (\mathcal{W}\oplus \text{span}\{v_1,\ldots, v_{i}\})$ by the mathematical induction.
}
Therefore, we have
\begin{align*}
     &\pr[\mathcal{W}\cap \mathcal{R}=\{0\}]\\
    &=\Big[\frac{q^{n_0} - q^{n_1}}{q^{n_0}}\Big]\cdot\Big[\frac{q^{n_0} - q^{n_1+1}}{q^{n_0}-q^1}\Big]
     \cdot \cdots \cdot \Big[\frac{q^{n_0} - q^{n_1+n_2-1}}{q^{n_0}-q^{n_2-1}}\Big] \\
    &=1 - O(q^{n_1+n_2-n_0-1}).    
\end{align*}
\end{proof}
\end{lemm}

\begin{lemm}    \label{lemm:max_zero_prob}
For any positive integer $n_\ell' > 3m_0$,
\begin{align}
 &\max_{x\neq \zerom{n_\ell'}{1}} \!\!
  \pr_{\vsec}\bparen*{ x^{\top} ((R_1^{\vsec})^{-1})^{\mathcal{A}} \! = \! \zerom{1}{\mm} } \!\leq\! \Big(\frac{n_\ell'\!-\!2\mm}{\Q}\Big)^{\!\mm}\!, \label{ineq:bit_max}\\
 &\max_{x\neq \zerom{n_\ell'}{1}} \!\!
 \pr_{\vsec}[ x^{\top} ([R_1^{\vsec}]_p^{-1})^{\mathcal{B}} \!=\! \zerom{1}{\mm} ] \leq \Big(\frac{n_\ell'\!-\!2\mm}{\Q}\Big)^{\!\mm}, \label{ineq:phase_max}
\end{align}
where 
the maximum is with respect to any nonzero vector $x\in\mathbb{F}_{q'}^{n_{\ell}}$,
and the random variable $\vsec=(\vsec_1,\ldots,\vsec_{4m_0})$ and the matrix $R_1^\vsec$ are defined in Section~\ref{sec:protocol}.
\end{lemm}
The proof of Lemma \ref{lemm:max_zero_prob} is given in Appendix \ref{sec:two_lemm_max}.

\subsection{The analysis of {protocol} after bit basis measurement} \label{subsec:Application}
Before we prove the upper bound \eqref{biterr} for the bit error probability,
we analyze the protocol when any bit basis state $|M\rangle_b\in \cH_{\mathrm{code}}^{(n_\ell)}$ is the input state of the code.
In the following, the parameter $(\mathcal{F}, S_{n_\ell})\in\zeta_{m_0,m_1}^{(n_\ell)} $ for the network operation is fixed but arbitrary.

In this case, the sender sends $\mathsf{E}^{(n_\ell)}_{R_e,R_s}(|M\rangle_{bb}\langle M|)$ over the network,
and the receiver receives the state $\Gamma[\mathcal{F}^{n_{\ell}},S_{n_\ell}]\circ\mathsf{E}^{(n_\ell)}_{R_e,R_s}(|M\rangle_{bb}\langle M|)$ on $\cH^{\otimes m_0\times n_\ell} = {(\cH')^{\otimes m_0\times n_\ell'}}$, where $\Gamma[\mathcal{F}^{n_{\ell}},S_{n_\ell}]$ is defined in \eqref{eq:network_operations}.
The receiver applies the decoder $\mathsf{D}^{(n_\ell)}_{R_s}$ and, finally, performs the bit basis measurement to the output state of the decoder.

Note that the bit basis measurement to the output state of the decoder commutes with the decoding operation $\mathsf{D}^{(n_\ell)}_{R_s}$.
That is, 
the process of applying the quantum decoder $\mathsf{D}^{(n_\ell)}_{R_s}$ and then performing the bit basis measurement on $\cH_{\mathrm{code}}^{(n_\ell)}$
is equivalent to
the process of performing the bit basis measurement on $(\cH')^{\otimes m_0\times n_\ell'}$ and then 
applying the classical decoding which corresponds to the quantum decoder $\mathsf{D}^{(n_\ell)}_{R_s}$.
Therefore, we adopt the latter method to calculate the bit error probability.

Let $Y\in \mathbb{F}_{\Q}^{m_0\times n_\ell'}$ be the outcome of the bit basis measurement on $(\cH')^{\otimes m_0\times n_\ell'} = \cHp_\mathcal{A}\otimes \cHp_\mathcal{B}\otimes \cHp_\mathcal{C}$.
From Eq.~\eqref{H2}, the matrix $Y$ is written as
\begin{align}
Y = \tilde{K}X'+ \tilde{W},     \label{eq:bit_output}
\end{align}
where 
$\tilde{K}\in\mathbb{F}_{\Q}^{m_0\times m_0}$ and $\tilde{W}\in\mathbb{F}_{\Q}^{m_0\times n_\ell'}$ are matrices equivalent to $K\in\mathbb{F}_{q}^{m_0\times m_0}$ and $WZ\in\mathbb{F}_{q}^{m_0\times n_\ell}$ in \eqref{H2} by field extension, respectively, 
and $X' := R_e X R_1^{\vsec}\in\mathbb{F}_{\Q}^{m_0\times n_\ell'}$ for $X\in\mathbb{F}_{\Q}^{m_0\times n_\ell'}$ defined with some matrices $\bar{E}_1 \in \mathbb{F}_{\Q}^{(m_0-m_1)\times \mm}$, $\bar{E}_2\in\mathbb{F}_{\Q}^{m_1\times \mm}$, and $\bar{E}_3\in\mathbb{F}_{\Q}^{m_1\times (n_\ell'\!-\!2\mm\!)}$ by
\begin{align}
    X &:= 
  \begin{bmatrix}
    \begin{bmatrix}
    \zerom{m_1}{\mm} \\
    \multicolumn{1}{c}{\multirow{2}{*}{$R_{2,b}$}} \\
    \multicolumn{1}{c}{} \\
    \end{bmatrix}
    ,
    \begin{bmatrix}
    \multicolumn{1}{c}{\multirow{2}{*}{$\bar{E}_1$}} \\
    \multicolumn{1}{c}{} \\
    \bar{E}_2 \\
    \end{bmatrix}
    ,
    \begin{bmatrix}
    \zerom{m_1}{n_\ell'\!-\!2\mm\!}\\
    M \\
    \bar{E}_3\\
    \end{bmatrix} 
    \end{bmatrix}. \label{eq:lb_me}
\end{align}

By Decode 1,
the matrix $Y$ is decoded as
\begin{align*}
Y_1 := Y(R_1^{\vsec})^{-1} = (\tilde{K} R_e X + \tilde{W}(R_1^{\vsec})^{-1}).
\end{align*}

Since the bit measurement outcome $O_b$ in Decode 2 is $Y_1^{\mathcal{A}} = (Y(R_1^{\vsec})^{-1})^\mathcal{A} = Y((R_1^{\vsec})^{-1})^\mathcal{A}$,
the equation \eqref{eq:Gauss_bit} is written as 
\begin{align}
 \!\!P_{b}D_b \!\paren*{\!\tilde{K}R_e \!\!\begin{bmatrix}\!\zerom{\!m_1\!}{\mm\!}\! \\ R_{2,b} \\ \end{bmatrix} 
                            \!\!+\! \tilde{W}((R_1^{\vsec})^{-1})^\mathcal{A} \!} 
 \!\!=\!\!\begin{bmatrix} \!\zerom{\!m_1\!}{\mm\!}\! \\ R_{2,b} \\   \end{bmatrix} \!.\!\!  \label{eq:Gauss_bit_2}
\end{align}
By Decode 2, the matrix $Y_1$ is decoded as 
\begin{align*}
Y_2 := D_b Y_1 = D_b(\tilde{K} R_eX + \tilde{W}(R_1^{\vsec})^{-1}).
\end{align*}

Though the decoding succeeds if $Y_2^{\mathcal{C}2}=M$, 
we evaluate instead the probability that
$P_b Y_2^{\mathcal{C}} = [\zerom{m_1}{n_\ell'-2m_0}^{\top}, M^{\top}, \bar{E}_3^{\top}]^{\top}$ holds.
In other words, since $P_bY_2^{\mathcal{C}}$ is written as 
\begin{align}
 P_{b}Y_2^\mathcal{C} 
 &= P_{b}D_b Y ((R_1^{\vsec})^{-1})^\mathcal{C} \nonumber \\
 &= P_{b} D_b 
 \paren*{\!\tilde{K}R_e \!\!\begin{bmatrix}\zerom{m_1}{n_\ell'\!-\!2\mm\!} \\ M \\ \bar{E}_3\end{bmatrix} 
                            + \tilde{W}((R_1^{\vsec})^{-1})^\mathcal{C} \!} ,
\end{align}
we evaluate the probability of
\begin{align}
\!\!\!
P_{b} D_b \!
 \paren*{\!\!\tilde{K}R_e \!\!\begin{bmatrix}\!\zerom{m_1}{n_{\!\ell}'\!-\!2\mm\!}\! \\ M \\ \bar{E}_3\end{bmatrix} 
                            \!\!+\! \tilde{W}((R_1^{\vsec})^{-1})^\mathcal{C}\!\! } 
 &\!\!= \!\!\begin{bmatrix}\!\zerom{m_1}{n_{\!\ell}'\!-\!2\mm\!} \!\\ M \\ \bar{E}_3\end{bmatrix}\!\!.\! \label{eq:C_recover}
\end{align}
Then, the decoding success probability is lower bounded by the probability that \eqref{eq:C_recover} holds.

\subsection{Upper bound of bit error probability} \label{subsec:prove_bit}

In this subsection, we derive the upper bound \eqref{biterr} for the bit error probability in Lemma \ref{lemm:errorp}.

Apply Lemma \ref{lemm:funda} to the following case:
\begin{gather}
\mathcal{V} := \mathbb{F}_{q'}^{m_0},\quad
\mathcal{W}_1 := \Ima \tilde{K} R_e |_{\mathcal{W}_b}, \label{eq:alloc}\\
\mathcal{W}_2 := \Ima \tilde{W},\quad
\mathcal{W}_3 := \mathcal{W}_b,\quad
A = (\tilde{K} R_e |_{\mathcal{W}_b})^{-1}\nonumber\\
[u_1+v_1,\ldots, u_{m_0}+v_{m_0}] := 
 \tilde{K}R_e \!\!\begin{bmatrix}\!\zerom{\!m_1\!}{\mm\!}\! \\ R_{2,b} \\ \end{bmatrix} 
                            \!\!+\! \tilde{W}((R_1^{\vsec})^{-1})^\mathcal{A},  \nonumber
\end{gather}
where $\mathcal{W}_b$ is the image of the projection $P_b$ defined in \eqref{eq:Gauss_bit}.
Let (A'), (B'), (C'), and (D') be the conditions (A), (B), (C), and (D) of Lemma \ref{lemm:funda} for this allocation, respectively.
If the conditions (A') and (B') hold,
the condition (C') implies that the equation \eqref{eq:Gauss_bit_2} has the solution $D_b$.
Moreover, it is clear from (D') that Eq. \eqref{eq:C_recover} holds, which implies there is no error in the protocol.
Therefore, we have the inequality
\begin{align}
\pr_{R_e,R_s} [\textnormal{(A')} \cap \textnormal{(B')}] \leq 1- \text{(bit error probability)}, \label{eq:bitlower}
\end{align}
where the probability of (A') depends on the random variable $R_e$ and that of (B') depends on random variables $R_e$ and $R_s=(\vsec, R_2)$.
That is, the evaluation of the bit error probability is reduced to the evaluation of the probability that both conditions (A') and (B') hold.

In the remainder of this subsection, we will prove the following lemma.
\begin{lemm}    \label{lemm:probab_ess}
The following inequalities holds:
\begin{align}
\pr_{R_e} [\textnormal{(A')}] & \geq 1-O\paren*{\frac{1}{\Q}}   \label{eq:pca},\\
\pr_{R_e,R_s} [\textnormal{(B')} | \textnormal{(A')} ]  & \geq 
    1-O\Big( \max \Big\{\frac{1}{\Q} , \frac{(n_\ell')^{m_0}}{(\Q)^{m_0-m_1}} \Big\} \Big).  \label{eq:pcb}
\end{align}
\end{lemm}

Then, by combining the inequality \eqref{eq:bitlower} with Lemma \ref{lemm:probab_ess}, we obtain the desired upper bound \eqref{biterr} for the bit error probability. 

\vspace{0.5em}

\subsubsection{Proof of lower bound \eqref{eq:pca} for $\pr_{R_e}[\textnormal{(A')}]$}   \label{sec:asdf1}
Apply Lemma \ref{lemm:space} to the case $\mathcal{V} := \mathbb{F}_{\Q}^{m_0}$, $\mathcal{W} := \Ima \tilde{W}$, and $\mathcal{R} := \Ima \tilde{K} R_e|_{\mathcal{W}_b}$.
In this case, 
we have $n_1=\rank \tilde{W} \leq \rank WZ \leq \rank W \leq m_a \leq m_1$
and 
$n_2=\rank \tilde{K} R_e|_{\mathcal{W}_b} = m_0-m_1$.
Therefore, Lemma \ref{lemm:space} implies the desired inequality \eqref{eq:pca}.

\vspace{0.5em}

\subsubsection{Proof of lower bound \eqref{eq:pcb} for $\pr_{R_e,R_s}[\textnormal{(B')}|\textnormal{(A')}]$}    \label{sec:asdf}
We derive the lower bound \eqref{eq:pcb} for $\pr_{R_e,R_s}[\textnormal{(B')}|\textnormal{(A')}]$, by three steps.
In the following, we assume the condition (A'). 

{\it Step 1:\enskip}
First, we give one necessary condition for (B') and calculate the probability that the necessary condition is satisfied. 
The condition (B') is equivalent to
\begin{align}
  &  \rank \Big( \tilde{K}R_e \begin{bmatrix}\zerom{m_1}{\mm} \\ R_{2,b} \\ \end{bmatrix} + \tilde{W}((R_1^{\vsec})^{-1})^\mathcal{A} \Big) \\
 &= \rank R_{2,b} + \rank \tilde{W}, \label{ineq:rank_ineq}
\end{align}
On the other hand, the following inequality holds from $\rank(A+B)\leq \rank A + \rank B$ and $\rank(AB)\leq \min\{\rank A,\rank B\}$ for any matrices $A$ and $B$:
\begin{align}
 & \rank \Big( \tilde{K}R_e \begin{bmatrix}\zerom{m_1}{\mm} \\ R_{2,b} \\ \end{bmatrix} + \tilde{W}((R_1^{\vsec})^{-1})^\mathcal{A} \Big) \\
 &\leq \rank R_{2,b} + \rank \tilde{W}((R_1^{\vsec})^{-1})^\mathcal{A} \nonumber \\
 &\leq \rank R_{2,b} + \rank \tilde{W}, \label{ineq:rank_ineq}
\end{align}
Therefore, the following condition is a necessary condition for (B'):
\begin{align}
\rank \tilde{W}((R_1^{\vsec})^{-1})^\mathcal{A} = \rank \tilde{W}. \label{cond:rank_err}
\end{align}

The condition (\ref{cond:rank_err}) holds if and only if $x^{\top}\tilde{W}((R_1^{\vsec})^{-1})^\mathcal{A}\neq \zerom{1}{m_0}$ holds for any $x \in \mathbb{F}_{\Q}^{m_0}$ such that $x^{\top}\tilde{W} \neq \zerom{n_\ell'}{1}$.
Apply Lemma \ref{lemm:max_zero_prob} to all $(\Q)^{\rank \tilde{W}}$ vectors in $\{x^{\top}\tilde{W} \neq \zerom{n_\ell'}{1} \mid x \in \mathbb{F}_{\Q}^{m_0}\}$, and then we have
\begin{align}
\pr_\vsec[\eqref{cond:rank_err}|\textnormal{(A')}]
 &\geq 1-(\Q)^{\rank \tilde{W}}\paren*{\frac{n_\ell'\!-\!2\mm\!}{\Q}}^{\mm} \nonumber  \\
 &\geq 1-(\Q)^{m_1}\paren*{\frac{n_\ell'\!-\!2\mm\!}{\Q}}^{\mm} \nonumber  \\
 &\geq 1- \frac{(n_\ell')^{\mm}}{(\Q)^{\mm-m_1}}.    \label{ineq:bound0}
\end{align}

{\it Step 2:\enskip}
In this step, we evaluate the conditional probability that (B') holds under the conditions (A') and \eqref{cond:rank_err}, i.e., $\pr_{R_e,R_s} [\textnormal{(B')} | \eqref{cond:rank_err}\cap\textnormal{(A')}]$.


Recall that 
the vectors $u_k,v_k\in\mathbb{F}_{\Q}^{m_0}$ for $k=1,\ldots,\mm$ are defined by \eqref{eq:alloc} as 
\begin{align*}
&[u_1 , \ldots , u_{\mm}] = \tilde{K}R_e \begin{bmatrix}\zerom{m_1}{\mm} \\  R_{2,b} \\  \end{bmatrix} ,\\
&[v_1 , \ldots , v_{\mm}] = \tilde{W}((R_1^{\vsec})^{-1})^\mathcal{A} .
\end{align*}
Let $m_2:= \rank R_{2,b}+\rank \tilde{W}$. 
Define an injective index function $i:\{1,...,m_0\}\to\{1,...,\mm\}$ such that $\rank (v_{i(1)}, \ldots,  v_{i(m_2)}) = \rank \tilde{W}$.
Note that the condition (B') holds 
if the $m_2$ vectors $u_{i(1)}+v_{i(1)},\ldots ,u_{i(m_2)}+v_{i(m_2)}$ are linearly independent.
Moreover, the condition (A') guarantees that the $m_2$ vectors $u_{i(1)}+v_{i(1)},\ldots ,u_{i(m_2)}+v_{i(m_2)}$ are linearly independent if the following condition holds:
\begin{align}
\mathcal{S}_u^{\perp} \cap \mathcal{S}_v^{\perp} &= \{\zerom{m_2}{1}\},
\label{cond:solution_empty}
\end{align}
where
\begin{align}
\mathcal{S}_u^{\perp} &:=
\px*{
x\in\mathbb{F}_{\Q}^{m_2} \ \Big\lvert\ 
[u_{i(1)}, \ldots,  u_{i(m_2)}]
x 
=  \zerom{m_0}{1}
},
\nonumber \\
\mathcal{S}_v^{\perp} &:=
\px*{
x\in\mathbb{F}_{\Q}^{m_2} \ \Big\lvert\ 
[v_{i(1)}, \ldots,  v_{i(m_2)}]
x
=  \zerom{m_0}{1}
}.
\nonumber
\end{align}
That is, we have the inequality
\begin{align}
\pr_{R_e,R_s}[\textnormal{(B')}|\eqref{cond:rank_err}\cap\textnormal{(A')}]   \geq 
\pr_{R_e,R_s}[\eqref{cond:solution_empty}| \eqref{cond:rank_err}\cap \textnormal{(A')} ]. \label{ineq:51B}
\end{align}

Then, we evaluate the probability that \eqref{cond:solution_empty} holds.
It follows from the definitions of vectors $u_1 , \ldots , u_{\mm},v_1 , \ldots , v_{\mm}$ and the index function $i$ that
\begin{align*}
\dim\mathcal{S}_u^{\perp} &\geq m_2-\rank[u_{i(1)}, \ldots,  u_{i(m_2)}]\geq \rank \tilde{W}, \\
\dim \mathcal{S}_v^{\perp} &= m_2-\rank[v_{i(1)}, \ldots,  v_{i(m_2)}] = \rank  R_{2,b}.
\end{align*}
This implies $\dim\mathcal{S}_u^{\perp} +\dim\mathcal{S}_v^{\perp}  \geq m_2$, and therefore \eqref{cond:solution_empty} holds only if 
\begin{align}
\dim \mathcal{S}_u^{\perp} = \rank \tilde{W}. \label{cond:last}
\end{align}
We calculate the conditional probability that \eqref{cond:solution_empty} holds by the following relation:
\begin{align}
 &   \pr_{R_e,R_s}[\eqref{cond:solution_empty}|  \eqref{cond:rank_err} \cap\textnormal{(A')}] \nonumber \\
 &= \pr_{R_e,R_s}[\eqref{cond:solution_empty} | \eqref{cond:last} \cap  \eqref{cond:rank_err} \cap \textnormal{(A')}]   \nonumber\\
 &  \quad    \cdot \pr_{R_e,R_s}[ \eqref{cond:last}  \cap  \eqref{cond:rank_err} \cap\textnormal{(A')}].
        \label{ineq:bound1}
\end{align}
Applying Lemma \ref{lemm:space} with $(n_0,\mathcal{W},\mathcal{R}):= (m_2,\mathcal{S}_v^{\perp}, \mathcal{S}_u^{\perp})$, we have
\begin{align}
\pr_{R_e,R_s}[\eqref{cond:solution_empty} |  \eqref{cond:last} \cap \eqref{cond:rank_err} \cap \textnormal{(A')}] = 1-O\paren*{\frac{1}{\Q}}.
        \label{ineq:bound2}
\end{align}
Moreover, the following inequality is proved in Appendix \ref{sec:lem_dim_WZ}:
\begin{align}
\pr_{R_e,R_s}[\eqref{cond:last} \!\cap\!  \eqref{cond:rank_err}\!\cap\! \textnormal{(A')} ]\!\geq \!1\!-\!O\paren*{\frac{1}{\Q}}. \label{lemm:dim_WZ}
\end{align}
Finally, combining the inequalities \eqref{ineq:51B}, \eqref{ineq:bound1}, \eqref{ineq:bound2}, and \eqref{lemm:dim_WZ}, we have the inequality 
\begin{align}
\pr_{R_e,R_s}[\textnormal{(B')}|\eqref{cond:rank_err}\cap\textnormal{(A')}]   
&\geq \pr_{R_e,R_s}[\eqref{cond:solution_empty}| \eqref{cond:rank_err}\cap \textnormal{(A')} ] \nonumber \\    
    & \geq 1-O\paren*{\frac{1}{\Q}}. 
        \label{ineq:bound3}
\end{align}


{\it Step 3:\enskip}
From the two inequalities \eqref{ineq:bound0} and \eqref{ineq:bound3}, the probability $\pr_{R_e,R_s}[\textnormal{(B')}|\textnormal{(A')}]$ is evaluated as  
\begin{align}
                        & \pr_{R_e,R_s}[\textnormal{(B')}|\textnormal{(A')}] \nonumber \\
                        &= \pr_{R_e,R_s}[\textnormal{(B')}\cap\eqref{cond:rank_err}|\textnormal{(A')}]   \nonumber \\
                        &= \pr_{R_e,R_s}[\textnormal{(B')}|\eqref{cond:rank_err}\cap\textnormal{(A')}]   
                            \cdot\pr_{R_e,R_s}[\eqref{cond:rank_err}|\textnormal{(A')}]   \nonumber \\
                        &\geq \paren*{1- O\paren*{\frac{1}{\Q}}}\paren*{1-\frac{(n_\ell')^{m_0}}{(\Q)^{m_0-m_1}}} \nonumber\\
                        &=1-O\Big( \max \Big\{\frac{1}{\Q} , \frac{(n_\ell')^{m_0}}{(\Q)^{m_0-m_1}} \Big\} \Big). \nonumber
\end{align}
Thus, we obtain the inequality \eqref{eq:pcb}.

\subsection{Phase error probability} \label{subsec:prove_phase}
Since Lemma \ref{lemm:invertible_to_unitary} implies that coding and node operations are considered as classical linear operations even in the phase basis,
we can apply similar analysis to the phase basis transmission as in Sections \ref{subsec:Application} and \ref{subsec:prove_bit}.

Consider the situation that any phase basis state $|M \rangle_p\in\cH_{\mathrm{code}}^{(n_\ell)}$ is encoded and transmitted through the quantum network.
In the same way as the bit basis states, we analyze the case that the receiver performs the phase basis measurement on $(\cH')^{\otimes m_0\times n_\ell'}$ first, and then applies the decoding operations.
After the phase basis measurement on  $(\cH')^{\otimes m_0\times n_\ell'}$,
the measurement outcome $Y\in\mathbb{F}_{q'}^{m_0\times n_\ell'}$ is written similarly to  \eqref{eq:bit_output} as
\begin{align*}
Y := [\tilde{K}R_e]_p Z [R_1^{\vsec}]_p+ \tilde{W}',  
\end{align*}
where $\tilde{W}'\in\mathbb{F}_{\Q}^{m_0\times n_\ell'}$ is a matrix such that $\rank\tilde{W}'\leq m_1$ and
\begin{align*}
    Z &:= 
  \begin{bmatrix}
    \begin{bmatrix}
    \bar{E}_1' \\
    \multicolumn{1}{c}{\multirow{2}{*}{$\bar{E}_2'$}} \\
    \multicolumn{1}{c}{} \\
    \end{bmatrix}
    ,
    \begin{bmatrix}
    \multicolumn{1}{c}{\multirow{2}{*}{$R_{2,p}$}}\\
    \multicolumn{1}{c}{}\\
    \zerom{\!m_1\!}{\mm\!} \\
    \end{bmatrix}
    ,
    \begin{bmatrix}
    \bar{E}_3'\\
    M \\
    \zerom{\!m_1\!}{n_\ell'\!-\!2\mm}\\
    \end{bmatrix}
  \end{bmatrix}
    \in \mathbb{F}_{\Q}^{m_0\times n_\ell'}
\end{align*}
for some matrices
$\bar{E}_1' \in \mathbb{F}_{\Q}^{m_1\times \mm}$, $\bar{E}_2'\in\mathbb{F}_{\Q}^{(m_0-m_1)\times \mm}$, and $\bar{E}_3'\in\mathbb{F}_{\Q}^{m_1\times (n_\ell'\!-\!2\mm\!)}$.
By the decoder, the matrix $Y$ is decoded as
\begin{align*}
Y_2 := [D_p]_p \paren*{[\tilde{K}R_e]_p Z + \tilde{W}' [(R_1^{\vsec})^{-1}]_p}.
\end{align*}

Consider applying Lemma \ref{lemm:funda} in the following case:
\begin{gather}
\mathcal{V} := \mathbb{F}_{q'}^{m_0},\quad
\mathcal{W}_1 := \Ima [\tilde{K}R_e]_p |_{\mathcal{W}_p},  \label{eq:alloc_2}\\
\mathcal{W}_2 := \Ima [\tilde{W}]_p,\quad
\mathcal{W}_3 := \mathcal{W}_p,\quad
A = ([\tilde{K} R_e]_p |_{\mathcal{W}_p})^{-1}\nonumber \\
[u_1\!+\!v_1,\ldots, u_{m_0}\!\!+\!v_{m_0}] := 
 [\tilde{K}R_e]_p \!\!\begin{bmatrix}R_{2,p} \\ \!\zerom{\!m_1\!}{\mm\!}\! \\ \end{bmatrix} 
                            \!\!+\! [\tilde{W}]_p[(R_1^{\vsec})^{-1}]_p^\mathcal{A},    \nonumber
\end{gather}
where $\mathcal{W}_p$ is the image of the projection $P_p$ defined in \eqref{eq:Gauss_bit}.
Let (A''), (B''), (C''), and (D'') be the conditions (A), (B), (C), and (D) of Lemma \ref{lemm:funda} for this allocation, respectively.
From Lemma \ref{lemm:funda}, if the conditions (A'') and (B'') hold,
there is no error in the protocol after the phase basis measurement. 
That is, we have the relation
\begin{align}
\pr_{R_e,R_s} [\textnormal{(A'')} \cap \textnormal{(B'')}] \leq 1- \text{(phase error probability)}. \label{eq:phaselower}
\end{align}
Moreover, 
by exactly the same way as in Sections \ref{sec:asdf1} and \ref{sec:asdf}, we have 
\begin{align}
\pr_{R_e} [\textnormal{(A'')}] & \geq 1-O\paren*{\frac{1}{\Q}}   \label{eq:pcaphase},\\
\pr_{R_e,R_s} [\textnormal{(B'')} | \textnormal{(A'')} ]  & \geq 
    1-O\Big(\! \max \Big\{\frac{1}{\Q} , \frac{(n_\ell')^{m_0}}{(\Q)^{m_0-m_1}} \Big\} \Big).   \label{eq:pcbphase}
\end{align}
Therefore, by combining inequalities \eqref{eq:phaselower}, \eqref{eq:pcaphase} and \eqref{eq:pcbphase}, we obtain the upper bound \eqref{phaseerr} of the phase  error probability in Lemma \ref{lemm:errorp}.

\section{Secure Quantum Network Code without Classical Communication} \label{sec:SQNC}
In the secure quantum network code given in Theorem \ref{theo:SQNC_protocol},
we assumed that the encoder and decoder share the negligible rate randomness $R_s$ secretly.
The secret shared randomness can be realized by secure communication.
The paper \cite{YSJL14} provided a secure classical communication protocol for the classical network as Proposition \ref{prop:secret_classical}.

\begin{prop}[{\cite[Theorem 1]{YSJL14}}]
Consider a classical network where each channel transmits an element of the finite field $\mathbb{F}_q$ and
each node performs a linear operation. 
Let the inequality $c_1 +c_2 < c_0$ holds 
for the transmission rate $c_0$ from Alice to Bob,
the rate $c_1$ of the noise injected by Eve, and
the rate $c_2$ of the information leakage to Eve.
For any positive integer $\beta$,
there exists a $k$-bit transmission protocol by $n_2:=k\beta c_0(c_0-c_2+1)$ uses of the network such that 
\begin{align*}
P_{\mathrm{err}} \leq k\frac{c_0}{q^{\beta c_0}} \text{  and  }
I(M;E) = 0,
\end{align*}
where $P_{\mathrm{err}}$ is the error probability and $I(M;E)$ is the mutual information between the message $M\in\mathbb{F}_2^k$ and the Eve's information $E$.
\label{prop:secret_classical}
\end{prop}

By attaching the protocol of Proposition \ref{prop:secret_classical} as a quantum protocol,
we can share the negligible rate randomness secretly as the following proof of Theorem \ref{theo:SQNC_protocol2}.

\begin{proof}[Proof of Theorem \ref{theo:SQNC_protocol2}]
Since the protocol of Proposition \ref{prop:secret_classical} can be implemented 
with the quantum network by sending bit basis states instead of classical bits, the following code satisfies the conditions of Theorem \ref{theo:SQNC_protocol2}.

In the same way as \eqref{def:block-length},
we choose $\alpha_\ell := \floor{5\log_q \ell}$, $n_{\ell,1}' := \floor{\ell/\alpha_\ell}$, $n_{\ell,1}:=\alpha n_{\ell,1}'$, $q':=q^{\alpha_\ell}$ for any sufficiently large $\ell$ such that $\alpha_\ell >0$ and $n_{\ell,1}' >3m_0$.
For the implementation of the code given in Section \ref{sec:protocol} with the block-length $n_{\ell,1}$ and the extended field of size $\Q$, 
the sender and receiver need to share 
the secret randomness which consists of $4m_0+2m_0(m_0-m_1)$ elements of $\mathbb{F}_{\Q}$.
Hence, using the protocol of Proposition \ref{prop:secret_classical}
with $(c_0,c_1,c_2):=(m_0,m_1,m_1)$, 
the sender secretly sends 
$k=\lceil (4m_0+2m_0(m_0-m_1))\log_2{\Q}\rceil$ bits to the receiver,
which is called the preparation protocol.
To guarantee that the error of the preparation protocol goes to zero, 
we choose $\beta = \lfloor 2\log_q \log_2 \ell \rfloor$.
Since $k$ is evaluated as 
$k=\lceil (4m_0+2m_0(m_0-m_1))\log_2{\Q}\rceil 
= \ceil{(4m_0+2m_0(m_0-m_1)) \floor{5\log_q \ell} \log_2 q}
\le \lceil 5 (4m_0+2m_0(m_0-m_1)) \log_2 \ell \rceil 
$,
we have $P_{\mathrm{err}} \le O(\log_2 \ell /(\log_2 \ell)^2) \to 0$.
Also, the preparation protocol requires 
$n_{\ell,2} = k \beta m_0(m_0-m_1+1)$ network uses.
Finally,  we apply the code given in 
Theorem \ref{theo:SQNC_protocol} with the block-length $n_{\ell,1}$ and the above chosen $\alpha_\ell$ and $\Q$.

The block-length of this code is $n_{\ell} = n_{\ell,1} +n_{\ell,2}$.
Since $n_{\ell,1} = \Theta(\ell)$ and 
\begin{align*}
n_{\ell,2} &\le  m_0(m_0-m_1+1)
\lceil 5 (4m_0+2m_0(m_0-m_1)) \log_2 \ell \rceil \\
 &\quad\cdot \lfloor 2\log_q \log_2 \ell \rfloor,
\end{align*}
we have $n_{\ell,2}/n_\ell\to 0$ and $n_{\ell,1}/n_\ell\to 1$.
Therefore,
Theorem~\ref{theo:SQNC_protocol} guarantees the conditions \eqref{AX} and \eqref{BX},
and this code do not assume any shared randomness, i.e, \eqref{CX} is satisfied.
Thus, this code realizes the required conditions.
\end{proof}

\section{Secrecy of our code} \label{sec:secrecy}

In this section, we show that the condition \eqref{eq:ent_fid_theo} in Theorem~\ref{theo:SQNC_protocol} and \eqref{BX} in Theorem \ref{theo:SQNC_protocol2}, i.e., 
$$
\lim_{\ell\to\infty} \max_{(\mathcal{F}, S_{n_{\ell}})} n_{\ell}(1-F_e^2(\rho_{\mathrm{mix}},\Lambda_{n_{\ell}}) ) = 0,
$$
guarantees the secrecy of the code.
The leaked information of a quantum protocol $\kappa$ is upper bounded by {\it entropy exchange} $H_e(\rho,\kappa) := H(\kappa\otimes \iota_R(\kb{\varphi}))=H(\kappa_E (\rho))$ as follows,
where $|\varphi\rangle$ is a purification of the state $\rho$,
{$\iota_{R}$ is the identity channel to the reference system,}
and $\kappa_E$ is the channel to the environment.
When the input state $\rho_x$ is generated subject to the distribution $p_x$,
the mutual information between the input system and the environment
is given as
$ H(\kappa_E (\sum_x p_x \rho_x))- \sum_x p_x H(\kappa_E (\rho_x))$,
which is upper bounded by $H_e(\kappa,\sum_x p_x \rho_x)$.
On the other hand,
the entropy exchange is upper bounded by the entanglement fidelity as \cite{Schumacher96}
\begin{align}
H_e(\rho,\kappa)  \leq h(F_e^{2}(\rho,\kappa)) + (1- F_e^2(\rho,\kappa))\log (d-1)^2,   \label{ineq:secret}
\end{align}
where $h(p)$ is the binary entropy defined as $h(p):= p\log p + (1-p)\log(1-p)$ for $0\leq p \leq 1$
and $d$ is the dimension of the input space of $\kappa$.
Hence, 
applying the inequality \eqref{ineq:secret} to an arbitrary averaged protocol $\Lambda_{n_\ell}$ and 
the completely mixed state $\rho_{\mathrm{mix}}$,
because $d=\dim \cH_{\mathrm{code}}^{(n_{\ell})} = O\big(q^{(m_0-2m_1)n_{\ell}}\big)$ in our code,
the condition \eqref{eq:ent_fid_theo} leads that
the entropy exchange of the averaged protocol is asymptotically $0$, i.e., there is no leakage in the averaged protocol.
Thus, the asymptotic correctability \eqref{eq:ent_fid_theo} also guarantees the secrecy of the code in Theorems \ref{theo:SQNC_protocol} and \ref{theo:SQNC_protocol2}.

\section{Conclusion} \label{sec:conclusion}

We have presented an asymptotically secret and correctable quantum network code as a quantum extension of the classical network codes given in \cite{Jaggi2008,HOKC17}.
To introduce our code, the network is constrained that the node operations are invertible linear operations to the basis states.
When the transmission rate of a given network is $m_0$ without attack and the maximum number of attacked channels is $m_1$,
by multiple uses of the network, 
our code achieves the rate $m_0-2m_1$ asymptotically without any classical communication.
Our code needs a negligible rate secret shared randomness but it is implemented by attaching a known secure classical network communication protocol \cite{YSJL14} to our quantum network code.
In the analysis of the code, we only considered the correctability because the secrecy is guaranteed by the correctness of the code protocol.
The correctability is derived analogously to the classical network codes \cite{Jaggi2008,HOKC17} but by evaluating the bit and phase error probabilities separately.


One remaining task is to show whether our code rate $m_0 - 2m_1$ is optimal or not. 
As a first step to discuss this problem, 
we may consider the quantum capacity when the network topology, node operations, and $m_1$ corrupted channels are fixed.
This problem is remained as a future study.


\section*{Acknowledgments}
SS is grateful to Yuuya Yoshida for helpful discussions and comments.
SS is supported by Rotary Yoneyama Memorial Master Course Scholarship (YM).
This work was supported in part by a JSPS Grant-in-Aids for 
Scientific Research (A) No.17H01280 and for Scientific Research (B) No.16KT0017, 
and Kayamori Foundation of Information Science Advancement.

\appendices
\section{Proof of Lemma \ref{lemm:invertible_to_unitary}} \label{sec:invertible_proof}
\begin{proof}[Proof of Lemma \ref{lemm:invertible_to_unitary}]
For any $x=(x_1,...,x_m),y=(y_1,...,y_m) \in \mathbb{F}_q^{m}$,
define an inner product 
\begin{align}
(x,y) := \sum_{i=1}^{m} \tr x_iy_i = \tr \sum_{i=1}^{m} x_iy_i, \label{def:inner_product}
\end{align}
where $\tr$ is defined in Section \ref{sec:notations}.
Let $T$ be a $m\times m$ matrix on $\mathbb{F}_q$.
If $x,y$ are considered as column vectors, it holds that $(Tx,y) = (x,T^{\top}y)$.
On the other hand, if $x,y$ are considered as row vectors, it holds that $(xT,y) = (x,yT^{\top})$.

First, we show $\lefto{A} | Z \rangle_p = | (A^{-1})^{\top} Z \rangle_p$ by considering $\mathbb{F}_q^{m}$ as a column vector space. For $\mathsf{L}^{(1)}(A) := \sum_{x\in\mathbb{F}_q^m} |Ax\rangle_{bb}\langle x|$ and $z\in \mathbb{F}_q^m$,
we have
\begin{align*}
 \mathsf{L}^{(1)}(A) | z \rangle_p   
&= \frac{1}{\sqrt{q^m}} \sum_{x\in\mathbb{F}_q^{m}} \omega^{-(x,z)} |Ax\rangle_b\\
                         &= \frac{1}{\sqrt{q^m}} \sum_{x'\in\mathbb{F}_q^{m}} \omega^{-(A^{-1}x', z)} |x'\rangle_b\\
                         &= \frac{1}{\sqrt{q^m}} \sum_{x'\in\mathbb{F}_q^{m}} \omega^{-(x', (A^{-1})^{\top}z)} |x'\rangle_b\\
                         &= |(A^{-1})^{\top} z\rangle_p.
\end{align*}
Since $\lefto{A} = \paren*{\mathsf{L}^{(1)}(A)}^{\otimes n}$,
we have $\lefto{A} | Z \rangle_p = | (A^{-1})^{\top} Z \rangle_p$.

Next, consider $\mathbb{F}_q^n$ as an $n$-dimensional row vector space over $\mathbb{F}_q$.
For $\mathsf{R}^{(1)}(B) := \sum_{x\in\mathbb{F}_q^n} |xB\rangle_{bb} \langle x|$ and $z\in \mathbb{F}_q^n$,
we have
\begin{align*}
\mathsf{R}^{(1)}(B) |z \rangle_p &= \frac{1}{\sqrt{q^n}} \sum_{x\in\mathbb{F}_q^{n}} \omega^{-(x, z)} |xB\rangle_b\\
                         &= \frac{1}{\sqrt{q^n}} \sum_{x''\in\mathbb{F}_q^{n}} \omega^{-(x''B^{-1}, z)} |x''\rangle_b\\
                         &= \frac{1}{\sqrt{q^n}} \sum_{x''\in\mathbb{F}_q^{n}} \omega^{-(x'', z(B^{-1})^{\top})} |x''\rangle_b\\
                         &= | z(B^{-1})^{\top}\rangle_p.
\end{align*}
Since $\righto{B} = \paren*{\mathsf{R}^{(1)}(B)}^{\otimes m}$,
we have $\righto{B} |Z\rangle_p = | Z(B^{-1})^{\top} \rangle_p$.
\end{proof}

\section{Proof of (\ref{eq:binary_PbPp})}
In this section, we show Lemmas \ref{lemm:maxen_pb} and \ref{lemm:PBPP} which shows the relationship between two maximally entangled states and projections $P_1,P_2$ defined by the bit and the phase bases.

Define the following maximally entangled states with respect to the bit and phase bases:
\begin{align*}
|\Phi_1\rangle := \frac{1}{\sqrt{q^m}}\sum_{i\in\mathbb{F}_q^m} |i, i\rangle_b, \quad
|\Phi_2\rangle := \frac{1}{\sqrt{q^m}}\sum_{z\in\mathbb{F}_q^m} |z, \bar{z}\rangle_p.
\end{align*}
We use the inner product $(\cdot,\cdot)$ defined in \eqref{def:inner_product} for the proofs.

\begin{lemm} \label{lemm:maxen_pb}
$|\Phi_1\rangle = |\Phi_2\rangle$.
\begin{proof}
The lemma is proved as follows:
\begin{align}
|\Phi_2\rangle  
                &= \frac{1}{\sqrt{q^m}}
                    \Big(\!\sum_{z\in\mathbb{F}_q^m}\!
                    \Big(\!\sum_{j\in\mathbb{F}_q^m}\! \frac{\omega^{-(z,j)}}{\sqrt{q^m}} |j\rangle_b\Big) 
               \otimes \Big(\!\sum_{l\in\mathbb{F}_q^m}\! \frac{\omega^{(z,l)}}{\sqrt{q^m}} |l\rangle_b\Big) \Big) \nonumber\\
                &= \frac{1}{\sqrt{q^m}}\sum_{z,j,l\in\mathbb{F}_q^m}\frac{\omega^{-(z,j-l)}}{q^m} |j,l\rangle_b \nonumber \\
                &= \frac{1}{\sqrt{q^m}}\sum_{j\in\mathbb{F}_q^m} |j,j\rangle_b = |\Phi_1\rangle, \label{eq:sum_over_phase}
\end{align}
where the first equality in \eqref{eq:sum_over_phase} holds because
\begin{align*}
\sum_{z\in\mathbb{F}_q^m}\frac{\omega^{-(z,j - l)}}{q^m} = 
    \begin{cases}
    0   &   \text{if } j \neq l,\\
    1   &   \text{otherwise}.
    \end{cases} 
\end{align*}
\end{proof}
\end{lemm}

From the above lemma, we denote $|\Phi\rangle := |\Phi_1\rangle =  |\Phi_2\rangle$.
Eq.~(\ref{eq:binary_PbPp}) is proved by the following lemma.
\begin{lemm}    \label{lemm:PBPP}
$P_1P_2 = P_2P_1 = \kb{\Phi}$.
\begin{proof}
The lemma is proved as follows:
\begin{align*}
P_1P_2 &= \sum_{i,z\in\mathbb{F}_q^m} {_b}\langle i,i|z,\bar{z}\rangle_p |i,i\rangle_{bp}\langle z,\bar{z}|\\
       &= \sum_{i,z\in\mathbb{F}_q^m} \frac{\omega^{-(z,i-i)}}{q^m} \sum_{j,l\in\mathbb{F}_q^m} \frac{\omega^{(z,j-l)}}{q^m} |i,i\rangle_{bb}\langle j,l|\\
       &= \sum_{i,j,l,z\in\mathbb{F}_q^m} \frac{\omega^{(z,j-l)}}{q^{2m}} |i,i\rangle_{bb}\langle j,l|\\
       &= \sum_{i,j\in\mathbb{F}_q^m} \frac{1}{q^m} |i,i\rangle_{bb}\langle j,j| = \kb{\Phi}. 
\end{align*}
\end{proof}
\end{lemm}

\section{Proof of Lemma \ref{lemm:max_zero_prob}}  \label{sec:two_lemm_max}

We use the following lemma \cite[Claim 5]{Jaggi2008} to prove Lemma~\ref{lemm:max_zero_prob}.
\begin{lemm}[{\cite[Claim 5]{Jaggi2008}}] \label{lemm:vandermonde}
Suppose independent $m$ random variables $\vsec_1,\ldots,\vsec_m\in\mathbb{F}_q$ are uniformly chosen in $\mathbb{F}_q$
and define the random matrix $Q\in \mathbb{F}_q^{l\times m}$ as $Q_{i,j}:=(\vsec_j)^i$.
For any row vectors $x\in\mathbb{F}_q^m$ and $y\in\mathbb{F}_q^{l}\backslash\{\zerom{1}{l}\}$ ($l\geq m$), we have
\begin{align}
\pr_\vsec[x=yQ] &\leq \Big(\frac{l}{q}\Big)^m. \label{ineq:vand}
\end{align}
%
\end{lemm}

Now, we prove Lemma \ref{lemm:max_zero_prob}.
\begin{proof}[Proof of Lemma \ref{lemm:max_zero_prob}]
Let $x=(x^\mathcal{A},x^\mathcal{B},x^\mathcal{C})\in \mathbb{F}_{\Q}^{\mm}\times\mathbb{F}_{\Q}^{\mm}\times\mathbb{F}_{\Q}^{n_\ell'\!-\!2\mm\!}$ be a nonzero row vector.
From the definition of $R_1^{\vsec}$, we have the relations
%
\begin{align}
&x ((R_1^{\vsec})^{-1})^{\mathcal{A}} = x^\mathcal{A}-x^\mathcal{B}(Q_3^{\top}+Q_4)-x^\mathcal{C}Q_1,  \label{eq111} \\
&x ([R_1^{\vsec}]_p^{-1})^{\mathcal{B}} = x^\mathcal{B}+x^\mathcal{A}(Q_4^{\top}+Q_3+Q_1^{\top}Q_2)+x^\mathcal{C}Q_2. \label{eq22222}
\end{align}

The inequality \eqref{ineq:bit_max} is proved as follows.
The relation \eqref{eq111} implies that the condition $x ((R_1^{\vsec})^{-1})^{\mathcal{A}} = \zerom{1}{\mm} $ holds in the following cases.
In each case, the probability for $x ((R_1^{\vsec})^{-1})^{\mathcal{A}} =  \zerom{1}{\mm}$ is calculated by Lemma \ref{lemm:vandermonde} as follows.
\begin{enumerate}
\item If $x^{\mathcal{C}} \neq \zerom{1}{n_\ell'\!-\!2\mm}$, the inequality \eqref{ineq:vand} for $Q:= Q_1$ implies 
\begin{align*}
\pr_\vsec[ x^\mathcal{A}-x^\mathcal{B}(Q_3^{\top}+Q_4) = x^\mathcal{C}Q_1] \leq \Big(\frac{n_\ell'\!-\!2\mm\!}{\Q}\Big)^{\mm}.
\end{align*}

\item If $x^\mathcal{B}\neq \zerom{1}{\mm}$ and $x^{\mathcal{C}} = \zerom{1}{n_\ell'\!-\!2\mm}$, 
the inequality \eqref{ineq:vand} for $Q:= Q_4$ implies 
\begin{align*}
&\pr_\vsec[ x^\mathcal{A}-x^\mathcal{B}Q_3^{\top} = x^\mathcal{B}Q_4 ] 
\leq \Big(\frac{\mm}{\Q}\Big)^{\mm}.
\end{align*}

\item If $x^\mathcal{A}\neq \zerom{1}{\mm}$, $x^\mathcal{B}=\zerom{1}{\mm}$, and $x^{\mathcal{C}} = \zerom{1}{n_\ell'\!-\!2\mm}$, 
the probability that \eqref{eq111} holds is zero.
\end{enumerate}
Since the inequality $n_\ell'> 3m_0$ holds from \eqref{cond:nprime}, we have 
\begin{align}
\Big(\frac{\mm}{\Q}\Big)^{\mm} < \Big(\frac{n_\ell'-2\mm}{\Q}\Big)^{\mm} .  \label{ineq:3m}
\end{align}
Therefore, we obtain the inequality \eqref{ineq:bit_max} in Lemma~\ref{lemm:max_zero_prob}.

Next, we show the inequality \eqref{ineq:phase_max} as follows.
The relation \eqref{eq22222} implies that 
the condition $x ([R_1^{\vsec}]_p^{-1})^{\mathcal{B}} =\zerom{1}{\mm}$ holds in the following cases.
In each case, the probability for $x ([R_1^{\vsec}]_p^{-1})^{\mathcal{B}} =\zerom{1}{\mm}$ is calculated by Lemma \ref{lemm:vandermonde} as follows.
\begin{enumerate}
\item If $x^{\mathcal{C}} \neq \zerom{1}{n_\ell'\!-\!2\mm\!}$, the inequality \eqref{ineq:vand} for $Q:= Q_2$ implies 
\begin{align*}
&\pr_\vsec[  x^\mathcal{B}+x^\mathcal{A}(Q_4^{\top}+Q_3+Q_1^{\top}Q_2)=-x^\mathcal{C}Q_2] \\
&\leq \Big(\frac{n_\ell'\!-\!2\mm\!}{\Q}\Big)^{\mm}.
\end{align*}

\item If $x^\mathcal{A}\neq \zerom{1}{\mm}$ and $x^{\mathcal{C}} = \zerom{1}{n_\ell'\!-\!2\mm\!}$, 
the inequality \eqref{ineq:vand} for $Q:= Q_3$ implies 
\begin{align*}
\pr_\vsec[  x^\mathcal{B}+x^\mathcal{A}(Q_4^{\top}+Q_1^{\top}Q_2)= -x^\mathcal{A}Q_3] 
\leq \Big(\frac{\mm}{\Q}\Big)^{\mm}.
\end{align*}

\item If $x^\mathcal{A}= \zerom{1}{\mm}$, $x^\mathcal{B}\neq\zerom{1}{\mm}$, and $x^{\mathcal{C}} = \zerom{1}{n_\ell'\!-\!2\mm\!}$, 
the probability that \eqref{eq22222} holds is zero.

\end{enumerate}
Therefore, from the inequality \eqref{ineq:3m}, we obtain the inequality \eqref{ineq:phase_max} in Lemma~\ref{lemm:max_zero_prob}.
\end{proof}

\section{Proof of \eqref{lemm:dim_WZ}} \label{sec:lem_dim_WZ}
From $\dim \mathcal{S}_u^{\perp} = m_2 - \rank [u_{i(1)}, \ldots,  u_{i(m_2)}]$, we have
\begin{align*}
 &\pr\!\big[\!\dim \!\mathcal{S}_u^{\perp} \!\!=\! \rank\! \tilde{W} \!\big]\! 
 \!=\! \pr\!\!\big[\!\rank [u_{i(1)}\!,\! \ldots\!, \! u_{i(m_2)}] \!=\! \rank\! R_{2\!,b}\big].
\end{align*}
Since $R_{2,b} = [u_{i(1)}, \ldots,  u_{i(m_0)}]$ is a random matrix with $\rank R_{2,b} =m_0-m_1$,
this probability is equivalent to
\begin{align*}
   & \pr\bparen*{\rank [u_{i(1)}, \ldots,  u_{i(m_2)}] = \rank R_{2,b}}\\
   &=  \pr\!\big[\!\rank [v_1,\ldots,v_{m_2}] = m_0-m_1 \big\lvert \nonumber   \\
     &\qquad\qquad\rank [v_1,\ldots,v_{m_0}]=m_0-m_1, v_k \in \mathbb{F}_{\Q}^{m_0-m_1}\big].\nonumber
\end{align*}
Therefore, it holds that
\begin{align}
   \lefteqn{\!\!\! \pr\bparen*{\rank [u_{i(1)}, \ldots,  u_{i(m_2)}] = \rank R_{2,b}}}\nonumber\\
 &\!\!\!\!\geq  \pr\bparen*{\rank [v_1,\ldots,v_{m_2}] = m_0-m_1  \big\lvert v_k \in \mathbb{F}_{\Q}^{m_0-m_1}}    \nonumber\\
 &\!\!\!\!\geq  \pr\bparen*{\rank [v_1,\ldots,v_{m_0\!-m_1}]\!=\!  m_0\!-\!m_1 \big\lvert v_k \in \mathbb{F}_{\Q}^{m_0-m_1}}.\!\! \label{prob:full_rank}
\end{align}
The probability (\ref{prob:full_rank}) is equivalent to the probability to choose $m_0-m_1$ independent vectors in $\mathbb{F}_{\Q}^{m_0-m_1}$:
\begin{align*}
     &\pr\bparen*{\rank [v_1,\ldots,v_{m_0-m_1}] = m_0-m_1 \big\lvert  v_k \in \mathbb{F}_{\Q}^{m_0-m_1}}\\
    &=\frac{(\Q)^{m_0\!-m_1}}{(\Q)^{m_0\!-m_1}}\!\cdot\!\frac{(\Q)^{m_0\!-m_1}\!\!-\!\Q}{(\Q)^{m_0\!-m_1}}\!
     \cdots\! 
     \frac{(\Q)^{m_0\!-m_1}\!\!-\!(\Q)^{m_0\!-\!m_1\!-\!1}}{(\Q)^{m_0-m_1}}\\
    &=1 - O\paren*{\frac{1}{\Q}}.  
\end{align*}
Therefore, \eqref{lemm:dim_WZ} holds with probability at least $1 - O(1/{\Q})$.

\end{document}